\documentclass[11pt]{article}
\usepackage[letterpaper,margin=.97in]{geometry}
\usepackage{amsmath, amssymb, amsthm, amsfonts}
\usepackage{centernot}
\usepackage{ifthen}
\usepackage{arrayjobx}
\usepackage{tikz}
\usetikzlibrary{positioning,decorations.pathreplacing}

\usepackage{cite}
\usepackage{times}
\usepackage{appendix}
\usepackage{graphicx}
\usepackage{color}
\usepackage[noend]{algpseudocode}
\usepackage{epstopdf}
\usepackage{wrapfig}
\usepackage[framemethod=tikz]{mdframed}
\usepackage[bottom]{footmisc}
\usepackage{enumitem}
\setitemize{noitemsep,topsep=0pt,parsep=0pt,partopsep=0pt}
\usepackage{caption}
\usepackage{xspace}
\usepackage{hyperref}
\hypersetup{
    unicode=false,          
    colorlinks=true,        
    linkcolor=red,          
    citecolor=green,        
    filecolor=magenta,      
    urlcolor=cyan           
}
\usepackage[capitalize]{cleveref}
\usepackage{fancyhdr} 
\usepackage{color} 
\usepackage[nofillcomment,linesnumbered,noend,noresetcount,noline]{algorithm2e}
\usepackage[noindentafter]{titlesec}
\DeclareMathOperator{\sgn}{sgn}
\titlespacing{\paragraph}{%
  0pt}{
  0.1\baselineskip}{
  1em}
\titlespacing\section{0pt}{8pt plus 1pt minus 1pt}{2pt plus 1pt minus 1pt}
\titlespacing\subsection{0pt}{8pt plus 1pt minus 1pt}{2pt plus 1pt minus 1pt}
\titlespacing\subsubsection{0pt}{8pt plus 1pt minus 1pt}{2pt plus 1pt minus 1pt}
\usepackage{graphicx} 
\usepackage{complexity} 
\usepackage{enumitem} 
\usepackage{shorttoc} 
\usepackage{array} 
\usepackage{float}
\usepackage{pdfsync}
\usepackage{verbatim}
\hypersetup{pageanchor=false}

\AtBeginDocument{%
 \abovedisplayskip=2pt plus 1pt minus 1pt
 \abovedisplayshortskip=0pt plus 1pt
 \belowdisplayskip=2pt plus 1pt minus 1pt
 \belowdisplayshortskip=0pt plus 1pt
}

\newtheoremstyle{slplain}
  {.4\baselineskip\@plus.1\baselineskip\@minus.1\baselineskip}
  {.3\baselineskip\@plus.1\baselineskip\@minus.1\baselineskip}
  {\itshape}
  {}
  {\bfseries}
  {.}
  { }
  {}
\theoremstyle{slplain} 

\algnewcommand\algorithmicswitch{\textbf{switch}}
\algnewcommand\algorithmiccase{\textbf{case}}

\algdef{SE}[SWITCH]{Switch}{EndSwitch}[1]{\algorithmicswitch\ #1\ \algorithmicdo}{\algorithmicend\ \algorithmicswitch}%
\algdef{SE}[CASE]{Case}{EndCase}[1]{\algorithmiccase\ #1}{\algorithmicend\ \algorithmiccase}%
\algtext*{EndSwitch}%
\algtext*{EndCase}%

\renewcommand{\paragraph}[1]{\vspace{0.07cm}\noindent {\bf #1}:}
\newtheorem*{theorem*}{Theorem}
\newtheorem{theorem}{Theorem}[section]
\newtheorem{lemma}[theorem]{Lemma}

\newtheorem{claim}[theorem]{Claim}
\newtheorem{corollary}[theorem]{Corollary}

\newtheorem{invariant}[theorem]{Invariant}

\makeatletter
\newtheorem*{rep@theorem}{\rep@title}
\newcommand{\newreptheorem}[2]{%
\newenvironment{rep#1}[1]{%
 \def\rep@title{#2 \ref{##1}}%
 \begin{rep@theorem}}%
 {\end{rep@theorem}}}
\makeatother

\newreptheorem{theorem}{Theorem}
\newreptheorem{lemma}{Lemma}
\newreptheorem{claim}{Claim}
\newreptheorem{invariant}{Invariant}
\newreptheorem{corollary}{Corollary}

\theoremstyle{definition}

\theoremstyle{remark}

\numberwithin{equation}{section}

\newtheoremstyle{etplain}
  {.0\baselineskip\@plus.1\baselineskip\@minus.1\baselineskip}
  {.0\baselineskip\@plus.1\baselineskip\@minus.1\baselineskip}
  {\itshape}
  {}
  {\bfseries}
  {.}
  { }
  {}

\newcommand{\idlow}[1]{\mathord{\mathcode`\-="702D\it #1\mathcode`\-="2200}}
\newcommand{\id}[1]{\ensuremath{\idlow{#1}}}
\newcommand{\litlow}[1]{\mathord{\mathcode`\-="702D\sf #1\mathcode`\-="2200}}
\newcommand{\lit}[1]{\ensuremath{\litlow{#1}}}

\newcommand{\namedref}[2]{\hyperref[#2]{#1~\ref*{#2}}}

\newcommand{\theoremref}[1]{\namedref{Theorem}{#1}}

\newcommand{\figureref}[1]{\namedref{Figure}{#1}}
\newcommand{\figurerefb}[2]{\hyperref[#1]{Figure~\ref*{#1}#2}}

\newcommand{\lemmaref}[1]{\namedref{Lemma}{#1}}
\newcommand{\claimref}[1]{\namedref{Claim}{#1}}
\newcommand{\invariantref}[1]{\namedref{Invariant}{#1}}

\newcommand{\equationref}[1]{\hyperref[#1]{(\ref*{#1})}}
\renewcommand{\eqref}{\equationref}

\newcommand{\confspacen}{\Lambda_n \to \mathbb{N}}
\newcommand{\confspacem}{\Lambda_m \to \mathbb{N}}

\newcommand{\reach}{\Longrightarrow}

\usepackage[textsize=tiny]{todonotes}

\newcommand{\DEBUG}[1]{}

\renewcommand{\setminus}{-}
\renewcommand{\emptyset}{\varnothing}

\newcommand{\EX}{\operatornamewithlimits{\mathbb{E}}}

\newcommand{\FullOrShort}{short}

\ifthenelse{\equal{\FullOrShort}{full}}{
        
  \newcommand{\fullOnly}[1]{#1}
  \newcommand{\shortOnly}[1]{}
  
  }{

    \newcommand{\fullOnly}[1]{}
    \newcommand{\shortOnly}[1]{#1}
    
  }

\usepackage{graphicx}
\usepackage{multirow}

\begin{document}


\date{}

\title{Time-Space Trade-offs in Population Protocols}

\author{
  \and
  Dan Alistarh\\
  \small ETH Zurich\\
	\small dan.alistarh@inf.ethz.ch\\
  \and
  James Aspnes\\
        \small Yale\\
        \small james.aspnes@yale.edu 
  \and
  David Eisenstat\\
        \small Google\\
        \small eisenstatdavid@gmail.com
  \and
  \and
  Rati Gelashvili\\
        \small MIT\\
	\small gelash@mit.edu
  \and
  Ronald L. Rivest\\
        \small MIT\\
	\small rivest@mit.edu
}

\maketitle
\begin{abstract}
{\small
Population protocols are a popular model of distributed computing, in which randomly-interacting agents with little computational power cooperate to jointly perform computational tasks. 
Inspired by developments in molecular computation, and in particular DNA computing, recent algorithmic work has focused on the complexity of solving simple yet fundamental tasks in the population model, such as \emph{leader election} 
(which requires stabilization to a single agent in a special ``leader'' state), and \emph{majority} (in which agents must stabilize to a decision as to which of two possible initial states had higher initial count). 
Known results point towards an inherent trade-off between the \emph{time complexity} of such algorithms, and the \emph{space complexity}, i.e. size of the memory available to each agent. 

 In this paper, we explore this trade-off and provide new upper and lower bounds for majority and leader election. 
 First, we prove a unified lower bound, which relates the space available per node with the time complexity achievable by a protocol:  
 for instance, our result implies that any protocol solving either of these tasks for $n$ agents using $O( \log \log n )$ states must take $\Omega( n / \polylog n )$ expected time. 
 This is the first result to characterize time complexity for protocols which employ super-constant number of states per node, and proves that fast, poly-logarithmic running times require protocols to have relatively large space costs. 
 
On the positive side, we give algorithms showing that fast, poly-logarithmic stabilization time can be achieved using $O( \log^2 n )$ space per node, in the case of both tasks. Overall, our results highlight a time complexity separation between $O(\log \log n)$ and $\Theta( \log^2 n )$ state space size for both majority and leader election in population protocols, and introduce new techniques, which should be applicable more broadly. 
 	 }	
\end{abstract}


\thispagestyle{empty}

\newpage
\setcounter{page}{1}

\section{Introduction}
Population protocols~\cite{AADFP06} are a model of distributed computing in which agents with little computational power and no control over the interaction schedule cooperate to collectively perform computational tasks. While initially introduced to model animal populations~\cite{AADFP06}, they have proved a useful abstraction for wireless sensor networks~\cite{PVV09, DV12}, chemical reaction networks~\cite{CCDS14}, and gene regulatory networks~\cite{BB04}. 
A parallel line of applied research has shown that population protocols can be implemented at the level of DNA molecules~\cite{CD13}, and that they are equivalent to computational tasks solved by living cells in order to function correctly~\cite{CCN12}. 

A population protocol consists of a set of $n$ finite-state agents, interacting 
  in pairs, where each interaction may update the local state of both participants. 
A \emph{configuration} captures a ``global state'' of the system at any given time, 
  and formally can be described by the counts of nodes in each state.      
The protocol starts in a valid initial configuration, 
  and defines the outcomes of pairwise interactions.
The goal is to have all agents stabilize to some configuration, 
  representing the output of the computation, 
  which satisfies some predicate over the initial configuration of the system. 
For example, 
one fundamental task is \emph{majority (consensus)}~\cite{AAE08, PVV09, DV12}, in which agents start in one of two input states $A$ and $B$, and must stabilize on a decision as to which state has a higher initial count.\footnote{In this paper, we will focus on the \emph{exact} majority task, in which the protocol must return the correct decision in all executions, as opposed to \emph{approximate} majority, where the wrong decision might be returned with low probability~\cite{AAE08}.} A complementary fundamental task is \emph{leader election}~\cite{AngluinAE08, AG15, DS15}, which requires the system to stabilize to final configurations in which a \emph{single} agent is in a special \emph{leader} state. 
  
The set of interactions occurring in each step is usually assumed to be decided by a \emph{probabilistic} scheduler, which picks the next pair to interact uniformly at random. 
One complexity measure is \emph{parallel time}, defined as the number of pairwise interactions until stabilization, divided by $n$, the number of agents. 
Another is \emph{space complexity}, defined as the number of \emph{distinct states} that an agent can represent internally. 

There has been considerable interest in the complexity of \emph{fundamental tasks} such as leader election and consensus in the population model. 
In particular, a progression of deep technical results~\cite{Doty14, CCDS14} has culminated in showing that \emph{leader election in sublinear time is impossible} for protocols which are restricted to a \emph{constant} number of states per node~\cite{DS15}.
At the same time, it is now known that leader election can be solved in $O(\log^3 n)$ time via a protocol requiring $O(\log^3 n)$ states per node~\cite{AG15}. 
For the majority task, the space-time complexity gap is much wider: sublinear time is impossible for protocols restricted to having at most \emph{four} states per node~\cite{AGV15}, but there exists a \emph{poly-logarithmic time} protocol which requires a number of states per node that is \emph{linear} in $n$~\cite{AGV15}. 
 
These results hint towards a trade-off between the \emph{running time} of a population protocol and the \emph{space}, or number of states, available at each agent. This relation is all the more important since time efficiency is critical in practical implementations, while technical constraints limit the number of states currently implementable in a molecule~\cite{CD13}. (One such technical constraint is the possibility of \emph{leaks}, i.e. spurious creation of states following an interaction~\cite{Leaks}. In DNA implementations, the more states a protocol implements, the higher the likelihood of a leak, and the higher the probability of divergence.) 
However, the characteristics of the time-space trade-off in population protocols are currently an open question. 

\paragraph{Contribution} In this paper, we take a step towards answering this question. First, we exhibit a general trade-off between the number of states available to a population protocol and its time complexity, which characterizes which deterministic predicates can be computed efficiently with limited space. Further, we give new and improved algorithms for majority and leader election, tight within poly-logarithmic factors. Our results, and their relation to previous work, are summarized in Figure~\ref{table}. 

\paragraph{Lower Bounds} When applied to majority, our lower bound  proves that there exist constants $c \in (0, 1)$ and $K \geq 1$ such that any protocol using $\lambda_n \leq c \log \log n$ states must take $\Omega( n / ( K^{\lambda_n} +  \epsilon n)^2 ) )$ time, where $\epsilon n$ is the difference between the initial counts of the two competing states. For example, any protocol using \emph{constant} states and supporting a constant initial difference necessarily takes \emph{linear} time. 

For leader election, we show that there exist constants $c \in (0, 1)$ and $K \geq 1$ such that any protocol using $ \lambda_n \leq c \log \log n$ states and electing at most $\ell(n)$ leaders, requires  $\Omega( n / ( K^{\lambda_n} \cdot \ell(n)^2 ) )$ expected time.  
Specifically, any protocol electing one leader using $\leq c \log \log n$ states requires $\Omega ( n / \textnormal{polylog } n )$ time.

\paragraph{Algorithms} On the positive side, we give new poly-logarithmic-time algorithms for majority and leader election which use $O( \log^2 n)$ space. 
Our majority algorithm, called Split-Join, runs in $O( \log^3 n )$ time both in expectation and with high probability, 
  and uses $O( \log^2 n)$ states per node. 
The only previosly known algorithms to achieve sublinear time required $\Theta( n )$ states per node~\cite{AGV15}, or exponentially many states per node~\cite{Berenbrink16}. 
Further, we give a new leader election algorithm called Lottery-Election, which uses $O(\log^2 n)$ states, 
  and stabilizes in $O( \log^{5.3} n \log \log n)$ parallel time in expectation and 
  $O( \log^{6.3} n)$ parallel time with high probability. 
This reduces the state space size by a logarithmic factor 
  over the best known algorithm~\cite{AG15}, at the cost of a poly-logarithmic running time increase with respect to the $O(\log^3 n)$ bound of~\cite{AG15}. 
 
 A key improvement with respect to previous work is that these time-space bounds hold independently of the initial configuration: for instance, the AVC algorithm~\cite{AGV15} could stabilize in poly-logarithmic time and using poly-logarithmic space under a restricted  set of initial configurations, e.g. assuming that the discrepancy $\epsilon$ between the two initial states is large. 
 Our Split-Join algorithm can achieve this for worst-case initial configurations, i.e. $\epsilon = 1 / n$. 
 
\begin{figure}[t]
\centering
{\small
\begin{tabular}{|c|c|c|c|c|}
\hline
Problem & Type & Expected Time Bound & Number of States  & Reference \\ \hline
\multirow{4}{*}{\parbox{2.2cm}{Exact Majority $\epsilon = 1/n$}} 
 & Algorithm & $O( n \log n )$ & 4 & \cite{DV12, MNRS14}  \\ 
 & Algorithm & $O( \log^2 n )$ & $\Theta( n )$ & \cite{AGV15}  \\ 
 & Lower Bound & $\Omega( n )$ & $\leq 4$ & \cite{AGV15}  \\ 
 & Lower Bound & $\Omega( \log n )$ & any & \cite{AGV15}  \\ \hline
\multirow{2}{*}{Leader Election} & Algorithm & $O( \log^3 n )$ & $O(\log^3 n)$ & \cite{AG15}  \\ 
& Lower Bound & $\Omega (n)$ & $O(1)$ & \cite{DS15}  \\ \hline
{Exact Majority} & \multirow{2}{*}{Lower Bound} & \multirow{2}{*}{$\Omega( n / \polylog n )$} & \multirow{2}{*}{$< 1/ 2 \log \log n$} & \multirow{2}{*}{This paper} \\ 
{Leader Election} & &  & & \\ \hline
{Exact Majority} & Algorithm & $O( \log^3 n )$ & $O(\log^2 n)$ & This paper  \\ \hline 
{Leader Election} & Algorithm & $O( \log^{5.3} n \log \log n)$ & $O(\log^2 n)$ & This paper  \\ \hline
\end{tabular}
}
\caption{Summary of results and relation to previous work.}
\label{table}
\end{figure}

\paragraph{Techniques}
The core of the lower bound is a two-step argument. 
First,  we prove that a hypothetical algorithm which would stabilize faster than allowed by the lower bound may reach ``stable'' configurations\footnote{Roughly, a configuration is stable if it may not generate any new types of states.} in which certain low-count states can be ``erased,'' i.e., may disappear completely following a sequence of reactions.  
In the second step, we engineer examples where these low-count states are exactly the set of all possible leaders (in the case of leader election protocols), or a set of nodes whose state may sway the outcome of the majority computation (in the case of consensus). This implies a contradiction with the correctness of the computation.   
Technically, our argument employs the method of bounded differences to obtain a stronger version of the main density theorem of~\cite{Doty14}, 
and develops a new technical characterization of the stable states which can be reached by a protocol, which does not require constant bounds on state space size, generalizing upon~\cite{DS15}. 
The argument provides a unified analysis: the bounds for each task in turn are corollaries of the main theorem characterizing the existence of certain stable configurations. 

On the algorithmic side, we introduce a new \emph{synthetic coin} technique, which allows nodes to generate almost-uniform local coins within a \emph{constant} number of interactions, by exploiting the randomness in the scheduler, and in particular the properties of random walks on the hypercube. 
Synthetic coins are useful for instance by allowing nodes to estimate the total number of agents in the system, and may be of independent interest as a way of generating randomness in a constrained setting. 

The Split-Join majority protocol starts from the following idea: nodes can encode their output opinions and their relative strength as \emph{integer values}: the value is positive if the node supports a majority of $A$, and negative if the node supports a majority of $B$. The higher the absolute value, the higher the ``confidence" in the corresponding outcome. 
Whenever two nodes meet, they \emph{average} their values, rounding to integers. 
This template has been used successfully in~\cite{AGV15} to achieve consensus in poly-logarithmic time, but doing so requires a state space of size at least linear $n$, as all values between $1$ and $n$ may need to be represented. 
Here we introduce a new \emph{quantized} averaging technique, by which nodes represent their output estimates by encoding them as powers of two. Again, opinions are averaged on each interaction. We prove that our quantization preserves correctness, and allows for fast (poly-logarithmic) stabilization, while reducing the size of the state space almost exponentially. 

The Lottery-Election protocol follows the basic and common convention that every agent starts as a potential leader, 
  and whenever two leaders interact, one drops out of contention.
However, once only a constant number of potential leaders remain, 
  they take a long time to interact, implying super-linear stabilization time.
To overcome this problem,~\cite{AG15} introduced a propagation mechanism, 
  by which contenders compete by comparing their seeding, 
  and the nodes who drop out of contention assume the identity of their victor, 
  causing nodes still in contention but with lower seeding to drop out.
Here we employ synthetic coins to ``seed'' potential leaders randomly, 
  which lets us reduce the number of leaders at an accelerated rate.
This in turn reduces the maximum seeding that needs to be encoded,
  and hence the number of states required by the algorithm.

\paragraph{Implications} Our lower bound can be seen as bad news for algorithm designers, since it show that stabilization is slow even if the protocol implements a super-constant number of states per node. 
On the positive side, the achievable stabilization time improves quickly as the size of the state space nears the logarithmic threshold: in particular, fast, poly-logarithmic time can be achieved using poly-logarithmic space. 

It is interesting to note that previous work by Chatzigiannakis et al.~\cite{CMNPS11} identified $\Theta( \log \log n)$ as a space complexity  threshold in terms of the \emph{computational power} of population protocols, i.e. the set of predicates that such algorithms can compute. In particular, their results show that variants of such systems in which nodes are limited to $o( \log \log n )$ space per node are limited to only computing \emph{semilinear predicates}, whereas $O( \log  n )$ space is sufficient to compute general symmetric predicates. 
By contrast, we show a \emph{complexity} separation between algorithms which use $O( \log \log n )$ space per node, and algorithms employing $\Omega( \log n )$ space per node.

\section{Preliminaries}
\paragraph{Population Protocols} 
A population protocol is a set of $n \geq 2$ agents, or nodes,
  each executing as a deterministic state machine with states 
  from a finite set $\Lambda_n$, whose size might depend on $n$.
The agents do not have identifiers, 
  so two agents in the same state are identical and interchangeable. 
Thus, we can represent any set of agents simply 
  by the counts of agents in every state, which we call a \emph{configuration}.
Formally, a \emph{configuration} $c$ is a function 
  $c: \Lambda_n \to \mathbb{N}$, where $c(s)$ represents the 
  \emph{number of agents in state $s$ in configuration $c$}.
We let $|c|$ stand for the sum, over all states $s \in \Lambda_n$, of $c(s)$,
  which is the same as the total number of agents in configuration $c$.
For instance, if $c$ is a configuration of all agents in the system,
  then $c$ describes the global state of the system, and $|c| = n$. 

Further, we define the set $I_n$ of all allowed initial configurations of the protocol 
  for $n$ agents, a finite set of output symbols $O$, a transition function 
  $\delta_n : \Lambda_n \times \Lambda_n \rightarrow \Lambda_n \times \Lambda_n$, 
  and an output function $\gamma_n : \Lambda_n \rightarrow O$.
The system starts in one of the initial configurations $i_n \in I_n$ (clearly, $|i_n| = n$),
  and each agent keeps updating its local state following interactions with other agents, 
  according to the transition function $\delta_n$.
The execution proceeds in \emph{steps}, where in each step a new pair of agents 
  is selected uniformly at random from the set of all pairs. 
Each of the two chosen agents updates its state according to the function $\delta_n$. 


A population protocol $\mathcal{P}$ will be a sequence of protocols 
  $\mathcal{P}_2, \mathcal{P}_3, \ldots,$ where for each $n \geq 2$ we have 
  $\mathcal{P}_n = (\Lambda_n, I_n, \delta_n, \gamma_n)$,
  defining the protocol states, initial configurations, transitions and 
  output mapping, respectively, for $n$ agents.
We say that a protocol is \emph{monotonic} if, for all $i \geq 1$,  
  (1) the number of states cannot decrease for higher node count $i$, 
  i.e., $|\Lambda_i| \leq |\Lambda_{i + 1}| $, and 
  (2) if the state counts are the same, then the protocol is the same, 
  i.e. if $|\Lambda_i| = |\Lambda_j|$, then $\Lambda_i =\Lambda_j$,
  $I_i = I_j$, $\delta_i = \delta_j$ and $\gamma_i = \gamma_j$.

We say that a monotonic population protocol $\mathcal{P}$ is \emph{input-additive}, if for 
  all $i_n \in I_n$ with $|\Lambda_n| = |\Lambda_{2n}|$ (which by monotonicity implies 
  $\mathcal{P}_n = \mathcal{P}_{2n}$), it holds that $i_n + i_n \in I_{2n}$, i.e. doubling the number 
  of agents in each state still leads to a valid initial configuration in a system of $2n$ agents.  

A configuration $c'$ is \emph{reachable} from a configuration $c$, 
  denoted $c \reach c'$, if there exists a sequence of consecutive steps 
  (interactions from $\delta_n$ between pairs of agents) 
  leading from $c$ to $c'$.
If the transition sequence is $p$, we will also write $c \reach_p c'$.
We call a configuration $c$ the \emph{sum of configurations} $c_1$ and $c_2$ 
  and write $c = c_1 + c_2$, iff $c(s) = c_1(s) + c_2(s)$ for all states $s \in \Lambda_n$.

\paragraph{Leader Election} 
Fix an $n \geq 2$.
In the \emph{leader election} problem, all agents start in the same initial state $A_n$, 
  i.e. $I_n = \{ i_n \}$ with $i_n(A_n) = n$. 
The output set is $O = \{\id{Win}, \id{Lose}\}$. 

We say that a configuration $c$ \emph{has a single leader} if 
  there exists some state $s \in \Lambda_n$ with $\gamma_n(s) = \id{Win}$ and $c(s) = 1$, 
  such that for any other state $s' \neq s$, $c(s') > 0$ implies $\gamma_n(s') = \id{Lose}$.
A population protocol $\mathcal{P}_n$ solves leader election 
  within $r$ steps with probability $1 - \phi$, 
  if, with probability $1 - \phi$, any configuration $c$ reachable from $i_n$ by the protocol within $\geq r$ steps has a single leader.

\paragraph{The Majority Problem}
Fix an $n \geq 2$.
In the \emph{majority problem}, we have two initial states $A_n, B_n \in \Lambda_n$,
  and $I_n$ consists of all configurations where each of the $n$ agents 
  is either in $A_n$ or $B_n$.
The output set is $O = \{\id{Win}_A, \id{Win}_B\}$. 
Given an initial configuration $i_n \in I_n$ let $a = i_n(A_n)$ and $b = i_n(B_n)$ be the count of the two initial states,  
  and let $\epsilon n = | a - b | $ denote the initial relative advantage of the majority state. 

We say that a configuration $c$ \emph{correctly outputs the majority decision} for $i_n$,
  when for any state $s \in \Lambda_n$, 
  if $a > b$ then $c(s) > 0$ implies $\gamma_n(s) = \id{Win}_A$,
  and if $b > a$ then $c(s) > 0$ implies $\gamma_n(s) = \id{Win}_B$. (The output in case of an initial tie is arbitrary between the two.) 
A population protocol $\mathcal{P}_n$ solves the majority problem
  from $i_n$ within $\ell$ steps with probability $1 - \phi$,
  if for any configuration $c$ reachable from $i_n$ by the protocol with $\geq \ell$ steps, 
  with probability $1 - \phi$, the protocol correctly outputs the majority for $i_n$. 
In this paper we consider the \emph{exact} majority task, 
  as opposed to \emph{approximate} majority~\cite{AAE08}, 
  which allows the wrong output with some probability. 

Finally, a population protocol $\mathcal{P}$ solves a task if, for every $n \geq 2$, 
the protocol instantiation $\mathcal{P}_n$ solves the task. The complexity measures for a protocol (defined below), are functions of $n$.

\paragraph{Stabilization}
A configuration $c$ of $n$ agents has a \emph{stable leader}, 
  if for all $c'$ reachable from $c$, it holds that $c'$ has a single leader.
Analogously, a configuration $c$ has a \emph{stable correct majority decision} 
  for $i_n$, if for all $c'$ with $c \reach c'$, 
  $c'$ correctly outputs the majority decision for $i_n$.
We say that a protocol \emph{stabilizes} when it reaches 
  such a stable output configuration.

\paragraph{Complexity measures} 
The above setup considers sequential interactions; 
however, interactions between pairs of distinct agents are independent, 
  and are usually considered as occurring in parallel. 
It is customary to define one unit of \emph{parallel time} as $n$ consecutive steps 
  of the protocol. 
We are interested in the expected parallel time for a population protocol to stabilize,
  i.e. the total number of (sequential) interactions from the initial configuration 
  required to reach a configuration with a stable leader 
  or with a stable correct majority decision, divided by $n$.
We also refer to this quantity as the \emph{stabilization time}.
If the expected time is finite, then we say that population protocol stably elects a leader 
  (or stably computes majority decision).
The \emph{space complexity} of a protocol is the number of states 
  in the protocol with respect to $n$, i.e. $| \Lambda_n |$.


\section{Lower Bound}
Our definition of population protocols allows the protocol to arbitrarily change its structure 
  as the number of states changes.
This generality strengthens the lower bound, and to our knowledge, 
  the definition covers all known population protocols.

\paragraph{Technical machinery}
We now lay the groundwork for our main argument by proving a set of technical lemmas.
In the following, we assume $n$ to be fixed; $\Lambda_n$ is the set of all states of the protocol. 
Let $S$ be an arbitrary set of states.
Abusing notation, we denote by $\delta_n ( S, S )$ the set of all states which can be generated 
  by transitions between states in $S$.
Then, we let $S_0$ be the set of states in the initial configurations of the protocol  
  and for all integers $k \geq 1$, we define inductively the set of states 
  $S_k = S_{k - 1} \cup \delta_n ( S_{k - 1}, S_{k - 1} ).$
Simply put, $S_k$ is the set of all states which can be generated 
  by the protocol after $k$ steps. 

Assume without loss of generality that all states in $\Lambda_n$ actually occur in some configurations 
  reachable by the protocol $\mathcal{P}_n$ from some initial configuration.
Then, it holds that $S_{|\Lambda_n| - 1} = S_{|\Lambda_n|} = \ldots,$ and $S_{|\Lambda_n| - 1} = \Lambda_n$.
We say that a configuration $c$ is $X$-rich with respect to a set of states $S$ if in the configuration $c$ all states in $S$ have count $\geq X$. 
A configuration is \emph{dense} with respect to a set of states $S$ if it is $n / M$-rich with respect to $S$, for some constant $M > 1$. 
We call an initial configuration \emph{fully dense}, 
  if it is dense with respect to the set of initial states $S_0$: practically, such configurations contain at least $n/M$ agents in each of the initial states.
We now prove the following statement, which says that all protocols with a limited state count will end up in a configuration in which all reachable states are present in large count. 
This lemma generalizes the main result of~\cite{Doty14} to a super-constant state space. 

\begin{replemma}{lem:density}
For all population protocols $A$ using $|\Lambda_n| \leq 1/2 \log \log n$ states,
  starting in a fully dense initial configuration, with probability $\geq 1 - (1 / n)^{0.99}$,
  there exists a step $j$ such that the configuration reached after $j$ steps is $n^{0.99}$-rich with respect to $\Lambda_n$. 
\end{replemma}
Fix a function $f : \mathbb{N} \to \mathbb{R}^{+}$. 
Fix a configuration $c$, and states $r_1$ and $r_2$ in $c$. 
A transition $\alpha : (r_1, r_2) \to (z_1, z_2)$ is an $f$\emph{-bottleneck} for $c$,
  if $c(r_1) \cdot c(r_2) \leq f(|c|)$. 
This bottleneck transition implies that the probability of a transition $(r_1, r_2) \to (z_1, z_2)$ is bounded. 
Hence, proving that transition sequences from initial configuration to final configurations 
  contain a bottleneck implies a lower bound on the stabilization time.
Conversely, the next lemma shows that, if a protocol stabilizes fast, then it must be possible 
  to stabilize using a transition sequence which does not contain any bottleneck. 
\begin{replemma}{lem:bottlefree}
Let $\mathcal{P}$ be a population protocol with $|\Lambda_n| \leq 1/2\log{\log{n}}$ states,
  and let $D_n \subseteq I_n$ be a non-empty set of fully dense initial configurations. 
Fix a function $f$. 
Assume that for sufficiently large $n$, $\mathcal{P}$ stabilizes in expected time 
  $o\left(\frac{n}{f(n)|\Lambda_n|^2}\right)$ from all $i_n \in D_n$.
Then, for all sufficiently large $m \in \mathbb{N}$ there is a configuration $x_m$ with $m$ agents, 
  reachable from some $i \in D_m$ and a transition sequence $p_m$, such that: 
\begin{enumerate}
 \item $x_m(s) \geq m^{0.99}$ for all $s \in \Lambda_{m}$,
 \item $x_m \reach_{p_m} y_m$, where $y_m$ is a stable output configuration, and
 \item $p_m$ has no $f$-bottleneck.
\end{enumerate}
\end{replemma}
Hence, fast stabilization from a sufficiently rich configuration requires 
  the existence of a bottleneck-free transition sequence.
The next \emph{transition ordering lemma}, due to~\cite{CCDS14}, 
  proves a property of such a transition sequence:  
  there exists an order over all states whose counts decrease more than some set threshold such that, 
  for each of these states $d_j$, 
  the sequence contains at least a certain number of a specific transition that consumes $d_j$, 
  but does not consume or produce any states $d_1, \ldots, d_{j-1}$ that are earlier in the ordering.
\begin{replemma}{lem:ordering}
Fix $b \in \mathbb{N}$, and let $B = |\Lambda_n|^2 \cdot b + |\Lambda_n| \cdot b$. 
Let $x, y : \confspacen$ be configurations of $n$ agents,
  such that for all states $s \in \Lambda_n$ we have $ x(s) \geq B^2$ and $x \reach_q y$ via 
  a transition sequence $q$ without a $B^2$-bottleneck.
Define
\begin{equation*}
\Delta = \{d \in \Lambda_n \mid y(d) \leq b\}
\end{equation*}
to be the set of states whose count in configuration $y$ is at most $b$. 
Then there is an order $\Delta = \{d_1, d_2,\ldots, d_k\}$, such that,
  for all $j \in \{1, \ldots, k\}$, there is a transition $\alpha_j$ of the form 
  $(d_j, s_j) \rightarrow (o_j, o_j')$ with $s_j, o_j, o_j' \not\in \{d_1, \ldots, d_j\}$,
  and $\alpha_j$ occurs at least $b$ times in $q$.
\end{replemma}

\paragraph{The Lower Bound Argument}
Given a population protocol $\mathcal{P}_n$, a configuration $c : \confspacen$ 
  and a function $g: \mathbb{N} \to \mathbb{N}^{+}$, we define the sets 
  $\Gamma_g(c) = \{s \in \Lambda_n \mid c(s) > g(|c|)\}$ and 
  $\Delta_g(c) = \{s \in \Lambda_n \mid c(s) \leq g(|c|)\}$.
  Intuitively, $\Gamma_g(c)$ contains states above a certain count, while $\Delta_g(c)$ contains state below that count. 
Notice that $\Gamma_g(c) = \Lambda_n \setminus \Delta_g(c)$. 



The proof strategy is to first show that if a protocol stabilizes ``too fast,''
  then it can also reach configurations where all agents are in states in $\Gamma_g(c)$.
Recall that a configuration $c$ is defined as a function $\Lambda_n \to \mathbb{N}$.
Let $S \subseteq \Lambda_n$ be some subset of states such that 
  all agents in configuration $c$ are in states from $S$, 
  formally, $\{s \in \Lambda_n \mid c(s) > 0\} \subseteq S$.
For notational convenience, we will write $c_{> 0} \subseteq S$ to mean the same.
\begin{theorem}
\label{thm:surgery}
Let $\mathcal{P}$ be a monotonic input-additive population protocol 
  using $|\Lambda_n| \leq 1/2 \log{\log{n}}$ states.
Let $g: \mathbb{N} \to \mathbb{N}^{+}$ be a function such that $g(n) \geq 2^{|\Lambda_n|}$ for all $n$ and 
  $6^{|\Lambda_n|} \cdot |\Lambda_n|^2 \cdot g(n) = o(n^{0.99})$.

Suppose $\mathcal{P}$ stabilizes in 
  $o \left( \frac{n}{(6^{|\Lambda_n|} \cdot |\Lambda_n|^3 \cdot g(n))^2} \right)$
  time from any initial configuration,
  and that it has fully dense initial configurations for all sufficiently large $n$.
Then, for infinitely many $m$ with $|\Lambda_m| = \ldots = |\Lambda_{3m}|$, 
  there exists an initial configuration of $2m$ agents $i \in I_{2m}$ and 
  stable output configuration $y$ of $m$ agents, such that for any configuration $u$ that satisfies 
  the boundedness predicate $\mathcal{B}(m, y)$ below, it holds that $i + u \reach z$ where $z_{> 0} \subseteq \Gamma_g(y)$.

We say that a configuration $u$ satisfies the boundedness predicate $\mathcal{B}(m, y)$ if 1) it contains between $0$ and $m$ agents, 
  2) all agents in $u$ are in states from $\Delta_g(y)$, i.e. $u_{>0} \subseteq \Delta_g(y)$, 
  and 3) $y(s) + u(s) \leq g(m)$ for all states $s \in \Delta_g(y)$.
\end{theorem}
\begin{proof}
For simplicity, set $b(n) = (6^{|\Lambda_n|} + 2^{|\Lambda_n|}) \cdot g(n)$,
   $b_2(n) = |\Lambda_n|^2 \cdot b(n) + |\Lambda_n| \cdot b(n)$, and $f(n) = (b_2(n))^2$.
The theorem statement implies that the protocol stabilizes in 
  $o\left(\frac{n}{f(n) |\Lambda_n|^2}\right)$ time.
By~\lemmaref{lem:bottlefree}, for all sufficiently large $m$ we can find 
  configurations of $m$ agents $i_m, x_m, y : \confspacem$, such that:
\begin{itemize}
\item $i_m \in I_m$ is a fully dense initial configuration of $m$ agents. 
\item $i_m \reach x_m \reach_{p_m} y$, where $y$ is a stable final configuration, as desired, 
  and the transition sequence $p_m$ does not contain an $f$-bottleneck (i.e. a $(b_2)^2$-bottleneck).
\item $\forall s \in \Lambda_m: x_m(s) \geq b_2(m)$. (Here, we use the assumption on the function $g$.) 
\end{itemize}

Moreover, because $|\Lambda_n| \leq 1/2 \log{\log{n}}$ for sufficiently large $n$, for infinitely many $m$ 
  it also additionally holds that $|\Lambda_m| = |\Lambda_{m+1}| = \ldots = |\Lambda_{3m}|$ 
  (otherwise $|\Lambda_n|$ would grow at least logarithmically in $n$).
This, due to monotonicity implies that
  population protocols $\mathcal{P}_m, \mathcal{P}_{m+1}, \ldots, \mathcal{P}_{3m}$ are all the same. 

Consider any such $m$.
Then, we can invoke~\lemmaref{lem:ordering} with $x_m$, $y$, transition sequence $p_m$
  and parameter $b = b(m)$.
The definition of $\Delta$ in the lemma statement matches $\Delta_b(y)$, and $b_2$ matches $b_2(m)$.
Thus, we get an ordering of states $\Delta_b(y) = \{d_1, d_2, \ldots, d_k\}$ and a corresponding 
  sequence of transitions $\alpha_1, \alpha_2, \ldots, \alpha_k$, where each $\alpha_j$ is of the form
  $(d_j, s_j) \to (o_j, o_j')$ with $s_j, o_j, o_j' \not \in \{d_1, d_2, \ldots, d_j\}$.
Finally, each transition $\alpha_j$ occurs at least $b(m) = (6^{|\Lambda_m|} + 2^{|\Lambda_m|}) \cdot g(m)$ 
  times in $p_m$.

We will now perform a set of transformations on the transition sequence $p_m$, called \emph{surgeries}, 
  with the goal of converging to a desired type of configuration. 
The next two claims specify these transformations, which are similar to the surgeries used in~\cite{DS15}, 
  but with some key differences due to configuration $u$ and 
  the new general definitions of $\Gamma$ and $\Delta$. 
The proofs are provided in the appendix.
Configuration $u$ is defined as in the theorem statement.
For brevity, we use $\Gamma_g = \Gamma_g(y)$, $\Delta_g = \Delta_g(y)$, 
  $\Gamma_b = \Gamma_b(y)$ and $\Delta_b = \Delta_b(y)$.
\begin{repclaim}{clm:surgery1}
There exist configurations $e : \confspacem$ 
  and $z'$ with $z'_{>0}  \subseteq \Gamma_g$, such that $e + u + x_m \reach z'$.
Moreover, we have an upper bound on the counts of states in $e$:
  $\forall s \in \Lambda_m: e(s) \leq 2^{|\Lambda_m|} \cdot g(m)$.
\end{repclaim}
\noindent The configuration $e + u + x_m$ has at most 
  $2^{|\Lambda_m|} \cdot g(m) \cdot |\Lambda_m| + g(m) + m$ agents, which is less than $3m$ for sufficiently 
  large $m$.
The state transitions used here and everywhere below are from $\delta_m = \ldots = \delta_{3m}$.

For any configuration $e : \confspacem$, let $e^{\Delta}$ be its projection onto $\Delta$, 
  i.e. a configuration consisting of only the agents from $e$ in states $\Delta$.
We can define $e^{\Gamma}$ analogously. By definition, $e^{\Gamma} = e \setminus e^{\Delta}$.
\begin{repclaim}{clm:surgery2}
Let $e$ be any configuration satisfying $\forall s \in \Lambda_m: e(s) \leq 2^{|\Lambda_m|} \cdot g(m)$.
There exist configurations $p $ and $w$, such that 
  $p_{>0} \subseteq \Delta_b$, $w_{>0} \subseteq \Gamma_g$ and $p + x_m \reach p + w + e^{\Delta_g}$.
Moreover, for counts in $p$, we have that $\forall s \in \Lambda_m: p(s) \leq b(m)$
  and for counts in $w^{\Gamma_g}$, we have $\forall s \in \Gamma_g: w(s) \geq 2^{|\Lambda_m|} \cdot g(m)$.
\end{repclaim} 
Let our initial configuration $i$ be $i_m + i_m$, which because $i_m$ is fully dense and 
  because of input-additivity, must also be a fully dense initial configuration from $I_{2m}$.
Trivially, $i \reach x_m + x_m$.
Let us apply~\claimref{clm:surgery2} with $e$ as defined in~\claimref{clm:surgery1}, 
  but use one $x_m$ instead of $p$.
This is possible because $\forall s \in \Lambda_m: x(s) \geq b_2(m) \geq b(m) \geq p(s)$.
Hence, we get $x_m + x_m \reach x_m + w + e^{\Delta_g} = x_m + e + (w - e^{\Gamma_g})$.
The configuration $w - e^{\Gamma_g}$ is well-defined because both $w$ and $e^{\Gamma_g}$ contain agents
  in states in $\Gamma_g$, with each count in $w$ being larger or equal to the respective count in $e^{\Gamma_g}$,
  by the bounds from the claims.

Finally, by~\claimref{clm:surgery1}, we have $u + x_m + e + (w - e^{\Gamma_g}) \reach z' + (w - e^{\Gamma_g})$. 
We denote the resulting configuration (with all agents in states in $\Gamma_g$) by $z$, 
 and have $i \reach z$, as desired.
\end{proof}
\noindent The following lemma is a technical tool that critically relies on 
  our  definitions of $\Gamma$ and $\Delta$.
\begin{replemma}{lem:stability}
Consider a population protocol in a system with any fixed number of agents $n$, 
  and an arbitrary fixed function $h:\mathbb{N} \to \mathbb{N}^{+}$ such that $h(n) \geq 2^{|\Lambda_n|}$.
Let $\xi(n) = 2^{|\Lambda_n|}$.
For all configurations $c, c' : \confspacen$, such that 
  $c_{>0} \subseteq \Gamma_h(c) \subseteq \Gamma_{\xi}(c')$, 
  any state producible from $c$ is also producible from $c'$. 
Formally, for any state $s \in \Lambda_n$, 
  $c \reach y$ with $y(s) > 0$ implies $c' \reach y'$ with $y'(s) > 0$.
\end{replemma}
\noindent Now we can prove the lower bounds on majority and leader election.
\begin{corollary}
\label{crl:boundle}
Any monotonic population protocol with $|\Lambda_n| \leq 1/2 \log{\log{n}}$ states 
  for all sufficiently large number of agents $n$ that stably elects at least one and at most $\ell(n)$ leaders, 
  must take $\Omega \left( \frac{n}{144^{|\Lambda_n|} \cdot |\Lambda_n|^6 \cdot \ell(n)^2} \right)$ 
  expected parallel time to stabilize.
\end{corollary}
\begin{proof}
We set $g(n) = 2^{|\Lambda_n|} \cdot \ell(n)$.
In the initial configurations for leader election all agents are in the same starting state.
Hence, all initial configurations are fully dense, 
  and all monotonic population protocols for leader election have to be input-additive. 

Assume for contradiction that the protocol stabilizes in parallel time 
  $o \left( \frac{n}{144^{|\Lambda_n|} \cdot |\Lambda_n|^6 \cdot \ell(n)^2} \right)$.
For all $n$, $I_n$ contains the only initial fully dense configuration with $n$ agents 
  in the same initial state.
Using~\theoremref{thm:surgery} and setting $u$ to be a configuration of zero agents, 
  we can find infinitely many configurations $i$ and $z$ of $2m$ agents, such that 
  (1) $i \reach z$,
  (2) $i \in I_{2m}$,
  (3) $|\Lambda_m| = |\Lambda_{2m}|$ and by monotonicity, 
  the same protocol is used for all number of agents between $m$ and $2m$, 
  (4) $z_{>0} \subseteq \Gamma_g(y)$, i.e. all agents in $z$ are in states that each have counts 
    of at least $2^{|\Lambda_m|} \cdot \ell(m)$ in some stable output configuration $y$ (of $|y| = m$ agents).

Because $y$ is a stable output configuration of a protocol that elects at most $\ell(m)$ leaders,
  none of these states in $\Gamma_g(y)$ that are present in strictly larger counts 
  ($2^{|\Lambda_m|} \cdot \ell(m) > \ell(m)$) in $y$ and $z$ can be leader states
  (i.e. $\gamma_m = \gamma_{2m}$ maps these states to output $\id{Lose}$).
Therefore, the configuration $z$ does not contain a leader.
This is not sufficient for a contradiction, because a leader election protocol may well pass through a 
  leaderless configuration before stabilizing to a configuration with at most $\ell(m)$ leaders.
We prove below that any configuration reachable from $z$ must also have zero leaders.
This implies an infinite time on stabilization from a valid initial configuration $i$ (as $i \reach z$) 
  and completes the proof by contradiction.

If we could reach a configuration from $z$ with an agent in a leader state, then by~\lemmaref{lem:stability},
  from a configuration $c'$ that consists of $2^{|\Lambda_m|}$ agents in each of the states in $\Gamma_g(y)$,
  it is also possible to reach a configuration with a leader, let us say through transition sequence $q$.
Recall that the configuration $y$ contains at least $2^{|\Lambda_m|} \cdot \ell(m)$ 
  agents in each of these states in $\Gamma_g(y)$, hence there exist disjoint configurations
  $c_1' \subseteq y$, $c_2' \subseteq y$, etc, $\ldots, c_{\ell(m)}' \subseteq y$ contained in $y$ and 
  corresponding transition sequences $q_1, q_2, \ldots, q_{\ell(m)}$, 
  such that $q_j$ only affects agents in $c_j'$ and leads one of the agents in $c_j'$ to become a leader.
Configuration $y$ is a output configuration so it contains at least one leader agent already,
  which does not belong to any $c_j'$ because as mentioned above, 
  all agents in $c_j'$ are in non-leader states.
Therefore, it is possible to reach a configuration from $y$ with $\ell(m) + 1$ leaders
  via a transition sequence $q_1$ on the $c_1'$ component of $y$, 
  followed by $q_2$ on $c_2'$, etc, $q_{\ell(m)}$ on $c_{\ell(m)}'$, 
  contradicting that $y$ is a stable output configuration.
\end{proof}
\noindent The proof of the majority lower bound follows similarly, and is deferred to the Appendix.
\begin{repcorollary}{crl:boundmaj}
Any monotonic population protocol with $|\Lambda_n| \leq 1/2 \log{\log{n}}$ states 
  for all sufficiently large number of agents $n$ that stably computes correct majority decision 
  for initial configurations with majority advantage $\epsilon n$, must take 
  $\Omega \left( \frac{n}{36^{|\Lambda_n|} \cdot |\Lambda_n|^6 \cdot \max(2^{|\Lambda_n|},\epsilon n)^2} \right)$ 
  expected parallel time to stabilize.
\end{repcorollary}
\section{Synthetic Coin Flips}
\label{sec:synthetic}
The state transition rules in population protocols are deterministic,
  i.e. the interacting nodes do not have access to random coin flips.
In this section, we introduce a general technique that extracts randomness from the schedule 
  and after only constant parallel time, allows the interactions to rely on close-to-uniform synthetic coin flips.  
This turns out to be an useful gadget for designing efficient protocols.

Suppose that every node in the system has a boolean parameter \id{coin},
  initialized with zero.
This extra parameter can be maintained independently of the rest of the protocol, 
  and only doubles the state space. 
When agents $x$ and $y$ interact, they both \emph{flip} the values of their coins. 
Formally, $x'.\id{coin} \gets 1-x.\id{coin}$ and $y'.\id{coin} \gets 1 - y.\id{coin}$,
  and the update rule is fully symmetric.

The nodes can use the \id{coin} value of the interaction partner as 
  a random bit in a randomized algorithm.
Clearly, these bits are not independent or uniform.
However, we prove that with high probability
  the distribution of \id{coin} quickly becomes close to uniform and remains that way.
  We use the concentration properties of random walks on the hypercube, 
   analyzed previously in various other contexts, e.g.~\cite{AR16}. 
   We also note that a similar algorithm is used by Laurenti et al.~\cite{LCK16} to generate randomness in chemical reaction networks, although they do not prove convergence bounds. 
\begin{theorem}
\label{thm:coin}
For any $i \geq 0$, let $X_i$ be the number of coin values equal to one in the system after $i$ interactions. 
Fix interaction index $k \geq \alpha n$ for a fixed constant $\alpha \geq 2$. 
For all sufficiently large $n$, we have that 
  $\Pr[|X_k - n/2| \geq n/2^{4\alpha} ] \leq 2 \exp(-\alpha \sqrt{n} /8)$.
\end{theorem}
\begin{proof}
We label the nodes from $1$ to $n$, 
  and represent their coin values by a binary vector of size $n$. 
  Let $k_0 = k - \alpha n$, and fix the  vector $v_{0}$ representing the coin values of the nodes 
  after the interaction of index $k_0$.
For example, if $k_0 = 0$, we know $v$ is a zero vector, 
  because of the way the algorithm is initialized.
  
For $1 \leq t \leq \alpha n$, denote by $Y_t$ the pair of nodes that are flipped 
  during interaction $k_0 + t$.
Then, given $v_{0}$ and $Y_1, \ldots, Y_{\alpha n}$, $X_k$ can be computed deterministically.
Moreover, it is important to note that $Y_j$ are \emph{independent} random variables and that   
  changing any one $Y_j$ can only change the value of $X_k$ by at most $4$.
Hence, we can apply McDiarmid's inequality~\cite{mcdiarmid}, stated below.
\begin{claim}[McDiarmid's inequality]
Let $Y_1, \ldots, Y_m$ be independent random variables and let $X$ be a function
  $X = f(Y_1, \ldots, Y_m)$, such that changing variable $Y_j$ only changes 
  the function value by at most $c_j$.
Then, we have that 
  $$\Pr[|X - \EX[X]| \geq \epsilon] \leq 2 \cdot \exp\left(-\frac{2\epsilon^2}{\sum_{j = 1}^m c_j^2}\right).$$
\end{claim}
\noindent 
Returning to our argument, assume that the sum of coin values after interaction $k-\alpha n$ is fixed and represented by the vector $v_0$. 
In the above inequality, we set $X_k = f_{v_0}(Y_1, \ldots, Y_{\alpha n})$, $\epsilon = \alpha n^{3/4}$ 
and $c_j = 4$, for all $j$ from $1$ to $\alpha n$. We get that 
  $$\Pr[|X_{k} - \EX[X_{k}|] \geq \alpha n^{3/4}] \leq 2 \cdot \exp(-\alpha^2 \sqrt{n}/8).$$ 
Fixing $v_0$ also fixes the number of ones among coin values in the system at that moment,
  which we will denote by $x$, i.e. 
  $$x := \sum_{j=1}^n v_j (k_0)= X_{k-\alpha n}.$$
We then notice that the following claim holds, whose proof is deferred to the Appendix.

\begin{repclaim}{clm:coinexp}
$\EX[X_{i+m} \mid X_i = x] = n / 2  + ( 1- 4 / n )^m \cdot ( x - n / 2)$.
\end{repclaim}

\noindent By~\claimref{clm:coinexp} we have 
  $$\EX[X_k \mid X_{k-\alpha n} = x] = n/2 + (1 - 4/n)^{\alpha n} \cdot (x - n/2).$$

  \noindent Since $0 \leq x \leq n$ and $(1-4/n)^{\alpha n} \leq \exp(-4 \alpha)$,
  we have that 
  $$n/2 - n/2^{4\alpha + 1} \leq \EX[X_k \mid X_{k-\alpha n} = x] \leq n/2 + n/2^{4 \alpha + 1}.$$
  
\noindent For any fixed $v$, we can apply McDiarmid's inequality as above, and get an upper bound 
  on the probability that $X_k$ (given fixed $v_0$), 
  diverges from the expectation by at most $\alpha n^{3/4}$.  
But, as we just established, for any $v_0$, the expectation we get in the bound 
  will be at most $n/2^{4\alpha + 1}$ away from $n/2$.
Combining these and using that $n/2^{4\alpha + 1} \geq \alpha n^{3/4}$ for all sufficiently large $n$
  gives the desired bound.
\end{proof}

\paragraph{Approximate Counting} 
Synthetic coins can be used to estimate the number of agents in the system, as follows. 
Each node executes the coin-flipping protocol, and counts the number of consecutive $1$ flips it observes, until the first $0$. 
Each agent records the number of consecutive $1$ coin flips as its \emph{estimates}. The agents then exchange their estimates, always adopting the maximum estimate. It is easy to prove that the nodes will eventually stabilize to a number which is a constant-factor approximation of $\log n$, with high probability. 
This property is made precise in the proof of~\lemmaref{lem:stage1}.

\section{The Lottery Leader Election Algorithm}
\paragraph{Overview} 
We now present a leader election population protocol using $O( \log^2 n )$ states. 
The protocol is split conceptually into two stages: 
in the first \emph{lottery} stage, each agent generates an individual 
  \emph{payoff} by flipping synthetic coins.
In the second \emph{competition} stage, each agent starts as a potential leader, and proceeds 
  to compare a \emph{value} associated with local state against that of each interaction partner. 
Of any two interacting agents, the one with ``larger" value wins, and remains in contention, 
  while the other drops out and becomes a \emph{minion}. 
Importantly, a minion agent can no longer be a leader, 
  but carries the value of its victor in future interactions. 
Thus, if an agent still in contention meets a minion with higher value, 
  it will drop out and adopt the higher value. 
Similarly, a minion will always adopt the highest leader value it encounters. 
This minion-based propagation mechanism has been used previously, 
  where the value corresponded to interaction counts~\cite{AG15}; 
here, we employ a more complex value function, which requires less states to store.

We now present the algorithm in detail. 
Initially, all nodes start in the same state, which is determined by seven parameters: 
  $\id{coin}$, $\id{mode}$, $\id{payoff}$, $\id{level}$, $\id{counter}$, $\id{phase}$ and $\id{ones}$.
The $\id{coin}$ parameter stores the synthetic coin value; it is binary, initially $0$.
The agent can be in one of four modes: 
  $\lit{seeding}$ (preparing the random coin), 
  $\lit{lottery}$ (generating the payoff value), 
  $\lit{tournament}$ (competing), or 
  $\lit{minion}$ (out of contention). 

We fix a parameter $m$ such that $m \geq (10 \log{n})^2$.
The protocol will use $O( m )$ states per node.

\paragraph{Seeding Mode} 
All agents start in $\lit{seeding}$ mode, with $\id{payoff}$ and $\id{level}$ values $0$, 
  and $\id{counter}$ value $4$. 
The goal of the four-interaction seeding mode is for the synthetic coin implemented 
  by the $\id{coin}$ parameter to mix well, 
  generating values which are close to uniform random. 
In the first four interactions, each agent simply decreases its $\id{counter}$ value. 
Once this $\id{counter}$ reaches $0$, the agent moves to $\lit{lottery}$ mode. 
By the properties of synthetic coins, when agents finish $\lit{seeding}$, 
  they hold $0$ or $1$ values in roughly equal proportion. 

\paragraph{Lottery Mode}
In $\lit{lottery}$ mode, an agent starts counting in its own $\id{payoff}$ the number 
  of consecutive interactions until observing $0$ as the $\id{coin}$ value of an interaction partner,
  by incrementing $\id{payoff}$ upon observing $1$ coins. 
When the agent first meets a $0$, or if the agent reaches the maximum possible value that 
  $\id{payoff}$ can hold, set to $\sqrt{m}$, the agent $x$ finalizes its $\id{payoff}$, 
  and changes its $\id{mode}$ to $\lit{tournament}$.

\paragraph{Tournament Mode}
The goal of tournament mode is two-fold: to force agents to compete (by comparing states), 
  and to generate additional tie-breaking random values (via the $\id{level}$ variable). 

Agents start with $\id{level} = 0$ and repeatedly attempt to increase their level. 
Each agent $x$ keeps track of \emph{phases}, 
  consisting of $\Theta(\log{\id{payoff}})$ consecutive interactions. 
In each phase, if all coin values of interaction partners are $1$, 
  then the $x.\id{level}$ is incremented; otherwise, it stays the same.  
An agent which reaches the maximum possible $\id{level}$, set at $\sqrt{m} / \log m$, 
  remains in $\lit{tournament}$ mode, but stops increasing its level.

Phases can be implemented by the $\id{phase}$ parameter used as a counter, 
  and a boolean $\id{ones}$ parameter.
$\id{phase}$ starts with $0$, and is incremented every interaction
  until it is reset to $0$ when the phase ends.
$\id{ones}$ is set to $\lit{true}$ at the beginning of each phase
  and becomes $\lit{false}$ is coin value $0$ is encountered. 
As mentioned above, the phase consists of $\Theta(\log{\id{payoff}})$
  interactions, and the exact function will be provided later. 
   
When two agents $x$ and $y$ in $\lit{tournament}$ mode meet, they compare 
  $(x.\id{payoff}, x.\id{level}, x.\id{coin})$ and $(y.\id{payoff}, y.\id{level}, y.\id{coin})$. 
If the former is larger, then agent $x$ eliminates agent $y$ from the tournament, and vice versa. 
Practically, agent $y$ sets its $\id{mode}$ to $\lit{minion}$, and adopts the 
  payoff and level values of the other agent.  
Note that agents with higher lottery payoff always have priority; 
  if both $\id{payoff}$ and $\id{level}$ are equal, the $\id{coin}$ value is used as a tie-breaker.  

\paragraph{Minion Mode}
An agent in $\lit{minion}$ mode keeps a record of the maximum $.\id{payoff}, .\id{level}$ pair 
  ever encountered in any interaction in its own $\id{payoff}$ and $\id{level}$ parameters.
If $x.\id{mode} = \lit{minion}$ and $y.\id{mode} = \lit{tournament}$, and 
  $(x.\id{payoff}, x.\id{level}) > (y.\id{payoff}, y.\id{level})$, 
  then the agent in state $y$ will be eliminated from contention, and turned into a minion.
Intuitively, minions help leaders with high payoffs and levels
  to eliminate other contenders by spreading information.
Importantly, minions do not use the coin value as a tie-breaker (as this could lead 
  to a leader eliminating itself via a cycle of interactions). 

\paragraph{Analysis Overview}
The main intuition is that one agent with the highest lottery payoff eventually becomes the leader. 
This is an agent that manages to reach a high level, and turns other competitors into its minions, 
  that further propagate the information about the highest payoff and level through the system.

Only nodes with $\id{mode} = \lit{minion}$ are non-leaders, 
  and once a node becomes a minion it remains a minion.
Therefore, we first prove in~\lemmaref{lem:lecorrect} that not all nodes can become minions, 
  and if there are $n-1$ minions in the system, then there is a stable leader.
The proofs are provided in the appendix.

The number of possible states of an agent can be determined by multiplying the maximum different 
  values of state parameters, giving 
  $O(1) \cdot \sqrt{m} \cdot \frac{\sqrt{m}}{\log{m}} \cdot O(\log{m}) = O(m)$ as desired.
  
Next, we prove in~\lemmaref{lem:stage1} that with probability at least $1-O(1)/n^3$,
  after $O(n \log{n})$ interactions, all agents will be in the competition stage,
  that is, either in $\lit{tournament}$ or $\lit{minion}$ mode, 
  with maximum $\id{payoff}$ at least $\log{n} / 2$ and at most $9 \log{n}$,
  and at most $5 \log{n}$ agents will have this maximum $\id{payoff}$.

We set the phase size of an agent with $\id{payoff} = p$ to $4.2 (\log{p} + 1)$.
We then show in~\lemmaref{lem:stage2} that with probability at least $1-O(1)/n^3$,
  only one contender reaches level $\ell = 3 \log{n} / \log \log n$,
  and it takes at most $O(n \log^{5.3} \log \log n)$ interactions for a new level to be reached up to $\ell$. 

The above claims imply that with high probability (at least $1-O(1)/n^3$),
  the protocol elects a single leader within 
  $O(n \log n) + O(n \log^{5.3} \log \log n) \cdot \ell = O(n \log^{6.3}{n})$ interactions,
  that is, $O(\log^{6.3} n)$ parallel time.
Finally,~\lemmaref{lem:electexp} gives a $O(\log^{5.3}{n} \log\log n)$ bound 
  on expected parallel time until stabilization. 
\section{The Split-Join Majority Algorithm}

\paragraph{Overview} 
In this section, we present an algorithm for exact majority using $O(\log^2 n)$ states per node. 
The main idea behind the algorithm is that each node state corresponds to an \emph{integer value}:  
the \emph{sign} of the value corresponds to the node's opinion about the majority state---by convention, $A$ is positive, and $B$ is negative. 
To minimize the state space, we will devise a special representation for the integer values, where not all integers will be representable. 
Whenever two nodes meet, they modify their respective values following a sequence of simple operations. 
The intuition is that, on each interaction, nodes \emph{average} their corresponding values. Averaging ensures that the sign of the sum over all values in the system never changes, while, initially, this sign corresponds to the sign of the majority value. Thus, the crux of the analysis will be to show that all nodes stabilize to values of the same sign, and do so quickly. 
 
\paragraph{States}
We now describe the algorithm in detail. 
Each node state corresponds to a pair of integers $x$ and $y$, represented by $\langle x,y \rangle$, 
  whereby both integers are powers of two from 
  $x,y \in \{0, 1, 2, 2^2, \ldots, 2^{\lceil \log{n} \rceil}\}$, where $n$ is the number of nodes. 
  The \emph{value} of a state corresponding to a pair $\langle x,y \rangle$ is
  $\id{value}( \langle x,y \rangle ) = x - y$.

Nodes start in one of two special states. 
By convention, nodes starting in state $A$ have the initial pair 
  $\langle 2^{\lceil \log{n} \rceil}, 0 \rangle$, 
  and the nodes starting in state $B$ have the initial pair $\langle 0, 2^{\lceil \log{n} \rceil} \rangle$. (Here, we assume that nodes have an estimate of $n$.)
  We distinguish between two types of states: \emph{strong} states represent non-zero values, and will always satisfy 
  $x \neq y$ and $2 \min(x, y) \neq \max(x, y)$. 
   The \emph{weak} states are represented as pairs $\langle 0,0 \rangle^+$ and $\langle 0,0 \rangle^-$, corresponding to value $0$, leaning towards $A$ or $B$, respectively. 
We will refer to states and their corresponding value pairs interchangeably.
The output function $\gamma$ maps each state to the output based the sign of its value
  (treating $\langle 0,0 \rangle^+$ as positive and $\langle 0,0 \rangle^-$ as negative).

\paragraph{Interactions}
The algorithm, specified in~\figureref{fig:logmajo}, 
  consists of a set of simple deterministic update rules for the node state, to be applied on every interaction. 
In the pseudocode, we make the distinction between pairs $\langle x, y \rangle$, which correspond to states,
  and pairs $[x, y]$ corresponding to tuples of integer values.
The interaction rule between the states $\langle x_1, y_1 \rangle$ and $\langle x_2, y_2 \rangle$ 
  of two interacting nodes is described by the function \emph{interact}.
The states after the interaction are $\langle x_1', y_1' \rangle$ and $\langle x_2', y_2' \rangle$.
All nodes start in the designated initial states and 
  continue to interact according to the \lit{interact} rule. 
If both interacting states are weak, nothing changes (line~\ref{line:zeros}).
Otherwise, three elementary reactions, \emph{cancel}, \emph{join}, and \emph{split} are applied,
  in this order.
Each reaction takes four values $x_1, y_1, x_2, y_2$ and 
  returns (possibly updated) values $x_1', y_1', x_2', y_2'$.
  
  The \emph{cancel} reaction matches positive and negative powers of $2$ from the two 
  interaction partners. 
  The \emph{join} operation matches values of the same sign, attempting to create higher powers of two. 
  The \emph{split} reaction does the opposite, by breaking powers of two into smaller powers. 
  Please see Figure~\ref{fig:interact} for an illustration. 
  Before returning, the \lit{interact} procedure normalizes the two states to satisfy some simple well-formedness conditions. 
    Notice that all these operations preserve the sum of values corresponding to their inputs.

\paragraph{Correctness and Stabilization} 
The first observation is that the sum of values in the system is constant throughout the execution. 
By construction, the initial sum is of the majority sign; since the sum stays constant, the algorithm may not reach a state in which all nodes have an opinion corresponding to the initial minority. This guarantees correctness.  
The stabilization bound follows by carefully tracking the maximum value in the system, 
and showing that minority values get cancelled out and switch sign quickly.

\begin{reptheorem}{thm:majcorrect}
The Split-Join algorithm will never stabilize to the minority decision, 
  and is guaranteed to stabilize to the majority decision within $O( \log^3 n )$ parallel time, 
  both in expectation and w.h.p. 
\end{reptheorem}

\section{Conclusion}

We have studied the trade-off between time and space complexity in population protocols, 
and showed that a super-constant state space is necessary to obtain fast, poly-logarithmic stabilization time for both leader election and exact majority. On the positive side, we gave algorithms which achieve poly-logarithmic expected stabilization time using $O( \log^2 n )$ states per node for both tasks. 
Our findings are not great news for practitioners, as even small constant state counts are currently difficult to implement~\cite{CD13}. It is interesting to note how nature appears to have overcome this impossibility~\cite{CCN12}: 
algorithms solving majority at the cell level do so \emph{approximately}, 
allowing for a positive probability of error, using small constant states per node and stabilizing in poly-logarithmic time~\cite{AAE08}. 

We open several avenues for future work. The first  is to characterize the time-space trade-off between $\log \log n$ and $\log^2 n$ states. This question will likely require the development of analytic techniques parametrized by the number of states. 
A second direction is exploring the space-time trade-offs for \emph{approximately correct} algorithms. 
 
\begin{figure}[h]
\hrule
\DontPrintSemicolon
{\centering
{\scriptsize
\begin{algorithm}[H]
\SetKwInput{KwState}{State Space}
\KwState{\;
$\id{Strong States}=\{\langle x,y\rangle |x,y\in \{0, 1, 2, 2^2, \ldots, 2^{\lceil \log{n} \rceil}\},x \neq y,2 \cdot \min(x, y) \neq max(x, y) \}$, \;
$\id{Weak States} = \{ \langle 0, 0 \rangle^+, \langle 0, 0 \rangle^-\}$\;
}

\KwIn{States of two nodes, $\langle x_1, y_1 \rangle$ and $\langle x_2, y_2 \rangle$}
\KwOut{Updated states $\langle x_1', y_1' \rangle$ and $\langle x_2', y_2' \rangle$}

\SetKwInput{KwState}{Auxiliary Procedures}


$\id{Reduce}( u, v ) = \left\{ 
 \begin{array}{lll} 
  \textnormal{$[0,0]$} & \textnormal{ if $u = v$ } \\
  \textnormal{$[u-v,0]$} & \textnormal{ if $u = 2v$ } \\
  \textnormal{$[0,v-u]$} & \textnormal{ if $2u = v$ } \\
  \textnormal{$[u,v]$} & \textnormal{ otherwise.} \\
  \end{array} 
  \right. $
  
\textbf{procedure} $\lit{cancel}(x_1, y_1, x_2, y_2)$\;
{ 
\Indp
  $[x_1', y_2'] \gets \id{Reduce}( x_1, y_2 )$\;

  $[x_2', y_1'] \gets \id{Reduce}( x_2, y_1 )$\;
\Indm
}

\textbf{procedure} $\lit{join}(x_1, y_1, x_2, y_2)$\;
{ 
\Indp
  \lIf{($x_1 - y_1 > 0$ and $x_2 - y_2 > 0$ and $y_1 = y_2$)} 
      {  
        $y_1' \gets y_1 + y_2$ and
	$y_2' \gets 0$
      }
  \lElse 
      {
        $y_1' \gets y_1$ and
	$y_2' \gets y_2$
      }
  \lIf{($x_1 - y_1 < 0$ and $x_2 - y_2 < 0$ and $x_1 = x_2$)} 
      {  
        $x_1' \gets x_1 + x_2$ and
	$x_2' \gets 0$
      }
  \lElse 
      {
        $x_1' \gets x_1$ and
	$x_2' \gets x_2$
      }
\Indm
}

\textbf{procedure} $\lit{split}(\langle x_1, y_1 \rangle, \langle x_2, y_2 \rangle)$\;
{ 
\Indp
  \If{($x_1 - y_1 > 0$ or $x_2 - y_2 > 0$) \textbf{and} $\max(x_1, x_2) > 1$ \textbf{and} $\min(x_1, x_2) = 0$} 
      {  
        $x_1' \gets \max(x_1, x_2) / 2$ and
	$x_2' \gets \max(x_1, x_2) / 2$
      }
  \lElse 
      {
        $x_1' \gets x_1$ and
	$x_2' \gets x_2$
      }
  \If{($x_1 - y_1 < 0$ or $x_2 - y_2 < 0$) \textbf{and} $\max(y_1, y_2) > 1$ \textbf{and} $\min(y_1, y_2) = 0$} 
      {  
        $y_1' \gets \max(y_1, y_2) / 2$ and
	$y_2' \gets \max(y_1, y_2) / 2$
      }
  \lElse 
      {
        $y_1' \gets y_1$ and
	$y_2' \gets y_2$
      }
\Indm
}

\textbf{procedure} $\lit{normalize}(x, y, v)$\;
{
 \Indp
 $[\hat{x}, \hat{y}] \gets \id{Reduce}( x, y )$\;
 \If{$x = 0$ and $y = 0$} 
 { 
   \lIf{$v \geq 0$} {
     $\langle x', y' \rangle \gets \langle 0, 0 \rangle^+$
   }\lElse {
     $\langle x', y' \rangle \gets \langle 0, 0 \rangle^-$
   }
 }\lElse{
   $\langle x', y' \rangle \gets \langle x, y \rangle$
 }
 \Indm
}


\textbf{procedure} $\lit{interact}(\langle x_1, y_1 \rangle, \langle x_2, y_2 \rangle)$\;
{ 
\Indp
  \lIf{$x_1 = y_1 = x_2 = y_2 = 0$} {
    $[\langle x_1', y_1' \rangle, \langle x_2', y_2' \rangle] 
       \gets [\langle x_1, y_1 \rangle, \langle x_2, y_2 \rangle]$\nllabel{line:zeros}
  }\Else {
    $[\hat{x_1}, \hat{y_1}, \hat{x_2}, \hat{y_2}] \gets
    \lit{split}(\lit{join}(\lit{cancel}(x_1, y_1, x_2, y_2)))$\;
    $\langle x_1', y_1' \rangle \gets \lit{normalize}(\hat{x_1}, \hat{y_1}, \hat{x_2} - \hat{y_2})$\;
    $\langle x_2', y_2' \rangle \gets \lit{normalize}(\hat{x_2}, \hat{y_2}, \hat{x_1} - \hat{y_1})$\;
  }
\Indm
}
\end{algorithm}}}
\hrule
\caption{The state update rules for the Split-Join majority algorithm.}
\label{fig:logmajo}
\end{figure}

\begin{figure}
\centering
\includegraphics[scale = 0.65]{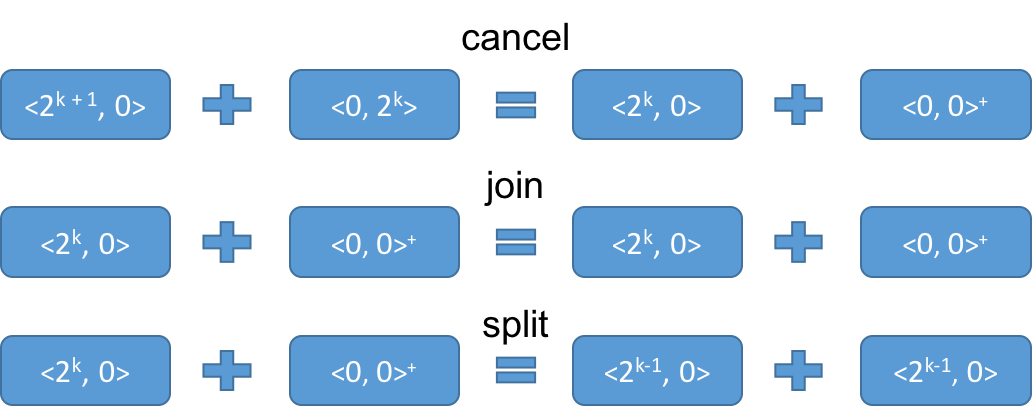}
\caption{Example of the outcome of the \emph{interaction} function. We apply the cancel, join, and split steps in sequence. In this case, the join step and the normalize step (not shown) are trivial.}
\label{fig:interact}
\end{figure}

\bibliographystyle{alpha}
\bibliography{biblio}

\appendix
\section{Lower Bound}
\begin{lemma}[Density Lemma]
\label{lem:density}
  For all population protocols $A$ using $|\Lambda_n| \leq 1/2 \log \log n$ states
  and starting in a fully dense initial configuration, with probability $\geq 1 - (1 / n)^{0.99}$, 
  there exists an integer $j$ such that the configuration reached after $j$ steps is $n^{0.99}$-rich with respect to $\Lambda_n$.
\end{lemma}
\begin{proof}
Recall that by definition, from a fully dense initial configuration
  every state in $\Lambda_n$ is producible.

We begin by defining, for integers $k \geq 0$, the function 

	$$f(k) = n 51^{-2^k + 1}.$$
	
\noindent Alternatively, we have that $f(k)^2 = f (k + 1) n / 51$. 

\noindent Let $c = 1 / 2$. Given the above, we notice that, with this choice it holds that, 
  for sufficiently large $n \geq 2$, 
\begin{itemize}
	\item $3 (c \log \log n)^2 / n \leq (1 / n)^{0.99},$ and
	
	\item for $0\leq k \leq c \log \log n$, we have that $f ( k ) \geq \max ( n^{0.99}, 50 \sqrt{n \log n}).$ 

\end{itemize}

We divide the execution into phases of index $k \geq 0$, each containing $n / 2$ consecutive interactions. 
For each $0 \leq k \leq |\Lambda_n| - 1$, we denote by $C_k$ the system configuration at the beginning of phase $k$. 

\paragraph{Inductive Claim.} We use probabilistic induction to prove the following claim: assuming that configuration $C_{k}$ is $f( k )$-rich with respect to the set of states $S_{k}$, with probability $1 - 3 |\Lambda_n| / n$, the configuration $C_{k + 1}$ is $f(k + 1)$-rich with respect to $S_{k + 1}$. 

For general $k \geq 0$, let us fix the interactions up to the beginning of phase $k$, and assume that configuration $C_k$ is $f(k)$-rich with respect to the set of states $S_k$.  Further, consider a state $q \in S_{k + 1}$. 
We will aim to prove that, with probability $1 - O( 1 / n)$, the configuration $C_{k + 1}$ contains state $q$ with count $\geq f(k+1)$. 

First, we define the following auxiliary notation. 
For any node $r$ and set of nodes $I$, count the number of interactions between $r$ and nodes in the set $I$, i.e. 
$$ \mathsf{intcount}( I, r ) = | \{ \textnormal{ interaction $j$ in phase $k$} : \textnormal{there exists $i \in I$ such that $e_j = (i, r)$ } \} |.$$

Next, we define the set of nodes in a state $s$ at the beginning of phase $k$ as 
$$ W(s) = \{ v : v \in V \textnormal{ and } C_k( v ) = s \}.$$
Finally, we isolate the set of nodes in state $s$ at the beginning of phase $k$ which \emph{did not interact} during phase $k$ as 
$$ W'(s) = \{ v : v \in W(s) \emph{ and } \mathsf{intcount}( V, v )  = 0 \}.$$ 

Returning to the proof, there are two possibilities for the state $q$. 
The first is when $q \in S_k$, that is, the state is already present at the beginning of phase $k$. 
But then, by assumption, state $q$ has count $\geq f( k )$ at the beginning of phase $k$. To lower bound its count at the end of phase $k$, it is sufficient to examine the size of the set $W'(q)$. For a node $v \in W(q)$, the probability that $v \in W'(q)$ is 
$$ \left(1 - \frac{1} {n}\right)^{n / 2} \geq 1 / 2,$$ 
\noindent by Bernoulli's inequality. Therefore the expected size of $W'(q)$ is at least $| W(q) | / 2$. Changing any interaction during phase $k$ may change $| W'(q) |$ by at most $1$, and therefore we can apply the method of bounded differences to obtain that 
$$ \Pr \left[ | W'(q) | < \frac{|W (q)|}{2} - \sqrt{n \log n} \right] \leq \exp\left( - \frac{ n \log n }{n}\right)  = \frac{1}{n}. $$ 

Since, by assumption, $| W(q) | \geq f( k ) \geq 10 \sqrt{n \log n}$, it follows that 
$$ \Pr \left[ |W'(q)| < \frac{2}{5} f(k) \right] \leq \frac{1}{n}.$$

Since $2 f(k) / 5 \geq f( k + 1)$, we have that $\Pr[ \# q( C_{k + 1}) \geq f(k + 1) ] \geq 1 - 1 / n,$ which concludes the proof of this case.  

 It remains to consider the case when $q \in S_{k + 1} \setminus S_k$. 
 Here, we know that there must exist states $q_i$ and $q_r$ in $S_{k}$ such that $\delta( q_i, q_r ) = q$. We wish to lower bound the number of interactions between nodes in state $q_i$ and nodes in state $q_r$ throughout phase $k$. To this end, we 
 isolate the set $R$ of nodes which are in state $q_r$ at the beginning of phase $k$, and only interact once during the phase, i.e. 
 $$R = \{ v: v \in W(q_r) \textnormal{ and } \mathsf{intcount}( V, v ) = 1   \},$$
 
and the set of nodes $R'$, which are in $R$, and only interacted once during phase $k$, with a node in the set $W'(q_i)$, i.e. 
 $$R' = \{ v: v \in R \textnormal{ and } \mathsf{intcount}( W'(q_i), v ) = 1 \}.$$

Notice that any node in the set $R'$ is necessarily in state $q$ at the end of phase $k + 1$. In the following, we lower bound the size of this set. 

First, a simple probabilistic argument yields that $E [ | R | ] \geq | W( q_r ) | / 4$. 
Since each interaction in this phase affects the size of $R$ by at most $2$ (since it changes the count of both interaction partners), 
we can again apply the method of bounded differences to obtain that 
$$ \Pr \left[ |R| < \frac{| W( q_r ) }{4} - 2 \sqrt{n \log n} \right] \leq \frac{1} {n}, $$ 

 \noindent implying that  
 $$ \Pr \left[ |R| < \frac{1}{20} f(k) \right]  \leq \frac{1}{n}.$$ 
 
To lower bound the size of $R'$, we apply again the method of bounded differences. 
We have that $|W'(q)| \geq (2 / 5) f(k)$, and that $|R| \geq (1 / 20) f(k)$, we have that 

 $$ \Pr\left[ |R'| \leq \frac{1}{50} \left( \frac{f(k - 1)^2}{n} \right) - \sqrt{n \log n}\right] \leq \frac{1}{n} .$$ 
 
 At the same time, we have that 
 $$ \frac{1}{50} \left(  \frac{f(k)^2 }{n}  \right) - \sqrt{n \log n} \geq \frac{51}{50} f(k + 1) - \frac{1}{50} f(k + 1) = f(k + 1),$$
\noindent which concludes the claim in this case as well.

\paragraph{Final Argument.} 
According to the lemma statement, we are considering an initial configuration in which all  initial states have count $\geq n / M$, for some constant $M \geq 0$. 
Let $k_0$ be the first positive integer such that $n / M \geq f(k_0)$. We have that  the initial configuration is $f(k_0)$-rich with respect to the set of initial states $S_0$. By a variant of the previous inductive claim, we obtain that, for any integer $0 \leq \ell \leq |\Lambda_n|$ satisfying $f( k_0 + \ell ) \geq \max( n^{0.99}, 10 \sqrt{ n \log n})$, at the beginning of phase $\ell$, configuration $C_\ell$ is $f( k_0 + \ell )$-rich with respect to $S_\ell$. 

It therefore follows that, with probability at least 
$$ (1 - 3|\Lambda_n|/n)^{|\Lambda_n|}\geq 1 - 3 (c\log \log n)^2 / n \geq 1 - 1 / n^{0.99},$$ there exists an integer $j$ such that the configuration reached after $j$ steps is $n^{0.99}$-rich with respect to $\Lambda_n$. 
\end{proof}

Given a protocol $\mathcal{P}_n$, for a configuration $c$ and a set of configurations $Y$, 
  let us define $T[c \reach Y]$ as the expected parallel time it takes from $c$
  to reach some configuration in $Y$ for the first time.
$\Pr[c \reach Y]$ stands for the probability of reaching a configuration in $Y$ from $c$. 
\begin{claim}
\label{clm:bottletime}
In a system of $n$ nodes, let $\gamma > 0, f : \mathbb{N} \to \mathbb{R}^{+}, c : \confspacen$, and 
  $X, Y$ be sets of configurations, such that $\Pr[c \reach X] \geq \gamma$, and every 
  transition sequence from every $x \in X$ to some $y \in Y$ has an $f$-bottleneck.
Then $T[c \reach Y] \geq \gamma \frac{n-1}{2 f(n) |\Lambda_n|^2}$.
\end{claim}
\begin{proof}
We will prove that for any $x \in X$, $T[x \reach Y] \geq \frac{n-1}{2 f(n) |\Lambda_n|^2}$ 
  holds, which implies the desired claim.
By definition, every transition sequence from $x$ to a configuration $y \in Y$ contains 
  an $f$-bottleneck, so it is sufficient to lower bound the expected time for the first 
  $f$-bottleneck transition to occur from $x$ before reaching $Y$.
In any configuration $c$ reachable from $x$, for any pair of states $r_1, r_2 \in \Lambda_n$ such that
  $r_1, r_2 \to p_1, p_2$ is a $f$-bottleneck transition in $c$, 
  the definition implies that $c(r_1) \cdot c(r_2) \leq f(n)$.
Thus the probability that the next pair of agents selected to interact are in states $r_1$ and $r_2$,
  is at most $\frac{2f(n)}{n(n-1)}$.
Taking an union bound over all $|\Lambda_n|^2$ possible such transitions, the probability that the 
  next transition is $f$-bottleneck is at most $|\Lambda_n|^2\frac{2f(n)}{n(n-1)}$.
Bounding by a geometric variable with success probability $\frac{2f(n)|\Lambda_n|^2}{n(n-1)}$,
  the expected number of interactions until the first $f$-bottleneck transition is at least 
  $\frac{n(n-1)}{2 f(n) |\Lambda_n|^2}$.
The expected parallel time is this quantity divided by $n$, completing the argument.
\end{proof}
\begin{lemma}
\label{lem:bottlefree}
Let $\mathcal{P}$ be a population protocol with $|\Lambda_n| \leq 1/2\log{\log{n}}$ states,
  and let $D_n \subseteq I_n$ be a non-empty set of fully dense initial configurations. 
Fix a function $f$. 
Assume that for sufficiently large $n$, $\mathcal{P}$ stabilizes in expected time 
  $o\left(\frac{n}{f(n)|\Lambda_n|^2}\right)$ from all $i_n \in D_n$.
Then, for all sufficiently large $m \in \mathbb{N}$ there is a configuration $x_m$ with $m$ agents, 
  reachable from some $i \in D_m$ and a transition sequence $p_m$, such that: 
\begin{enumerate}
 \item $x_m(s) \geq m^{0.99}$ for all $s \in \Lambda_{m}$,
 \item $x_m \reach_{p_m} y_m$, where $y_m$ is a stable output configuration, and
 \item $p_m$ has no $f$-bottleneck.
\end{enumerate}
\end{lemma}
\begin{proof}
$D_n$ is a set of some legal initial configurations for $n$ agents, 
  which are all given to be fully dense.
We know that the expected time to reach a stable output configuration from 
  these initial configurations is finite.
Hence if $i \reach x_m$ for $i \in D_m$, then a stable output configuration $y_m$ must be reachable 
  from $x_m$ through some transition sequence $p_m$, 
  but we also need $x_m$ and $p_m$ to satisfy the first and third requirements.

We know $|\Lambda_n| \leq 1/2 \log{\log{n}}$ for all large enough $n$.
Hence, by~\lemmaref{lem:density}, starting in any fully dense configuration $i_n \in D_n$, 
  with probability at least $1-(1/n)^{0.99}$, an $n^{0.99}$-rich configuration is reachable.
So for $n > 2$, we get that $\Pr[i_n \reach X_n] \geq 1/2$ where
  $X_n = \{x \mid i_n \reach x$ and $(\forall s \in \Lambda_n) x(s) \geq n^{0.99}\}$. 

Let $Y_n$ be a set of all stable output configurations with $n$ agents.
Suppose that every transition sequence from every configuration $x \in X_n$ to some $y \in Y_n$ 
  has an $f$-bottleneck.
Then, using~\claimref{clm:bottletime}, the expected time to stabilize from $i \in D_n$ is 
  $T[i_n \reach Y_n] \geq \frac{1}{2} \cdot \frac{n-1}{2 f(n) |\Lambda_n|^2} = \Theta(\frac{n}{f(n)|\Lambda_n|^2})$.
But we know that the protocol stabilizes from $i \in D_n$ in time $o(\frac{n}{f(n)|\Lambda_n|^2})$,
  implying that for all sufficiently large $m$, we can find $x_m \in X_m$ from which it is possible 
  to reach a stable output configuration in $Y_m$ without an $f$-bottleneck.
First requirement is satisfied by the definition of $X_m$,
  and we let $p_m$ be the transition sequence from $x_m$ to some $y_m \in Y_m$ without an $f$-bottleneck.
\end{proof}
\begin{lemma}
\label{lem:ordering}
Fix $b \in \mathbb{N}$, and let $B = |\Lambda_n|^2 \cdot b + |\Lambda_n| \cdot b$. 
Let $x, y : \confspacen$ be configurations of $n$ agents,
  such that for all states $s \in \Lambda_n$ we have $ x(s) \geq B^2$ and $x \reach_q y$ via 
  a transition sequence $q$ without a $B^2$-bottleneck.
Define
\begin{equation*}
\Delta = \{d \in \Lambda_n \mid y(d) \leq b\}
\end{equation*}
to be the set of states whose count in configuration $y$ is at most $b$. 
Then there is an order $\Delta = \{d_1, d_2,\ldots, d_k\}$, such that,
  for all $j \in \{1, \ldots, k\}$, there is a transition $\alpha_j$ of the form 
  $(d_j, s_j) \rightarrow (o_j, o_j')$ with $s_j, o_j, o_j' \not\in \{d_1, \ldots, d_j\}$,
  and $\alpha_j$ occurs at least $b$ times in $q$.
\end{lemma}
\begin{proof}
This part of the argument is identical to~\cite{CCDS14, DS15} and is described below for the sake of completeness.

Let $k = |\Delta|$ and define $\Delta_k = \Delta$.
We will construct the ordering in reverse, i.e. we will determine $d_j$ 
  for $j = k, k-1, \ldots, 1$ in this order.
At each step, we will define the next $\Delta_{j-1}$ as $\Delta_j \setminus \{d_j\}$.

We start by setting $j = k$.
For all $j$ we define $\Phi_j : (\confspacen) \to \mathbb{N}$ based on $\Delta_j$ as 
  $\Phi_j(c) = \sum_{d \in \Delta_j} c(d)$, i.e. the number of agents in states from $\Delta_j$ 
  in configuration $c$.
Notice that once $\Delta_j$ is well-defined, so is $\Phi_j$.

The following works for all $j \geq 1$ and lets us construct the ordering.
Because $y(d) \leq b$ for all states in $\Delta$, it follows that 
  $\Phi_j(y) \leq j \cdot b \leq |\Lambda_n| \cdot b$.
On the other hand, we know that $x(d) \geq B$ for all $d \in \Delta_j$, 
  hence $\Phi_j(x) \geq B \geq |\Lambda_n| \cdot b \geq \Phi_j(y)$.
Let $c'$ be the last configuration along $q$ from $x$ to $y$ where $\Phi_j(c') \geq B$,
  and $r$ be the suffix of $q$ after $c'$.
Then, $r$ must contain a subsequence of transitions $u$ each of which strictly decreases $\Phi_j$, 
  with the total decrease over all of $u$ being at least 
  $\Phi_j(c') - \Phi_j(y) \geq B - |\Lambda_n| \cdot b \geq |\Lambda_n|^2 \cdot b$.

Let $\alpha: r_1, r_2 \to p_1, p_2$ be any transition in $u$.
$\alpha$ is in $u$ so it strictly decreases $\Phi_j$, and without loss of generality $r_1 \in \Delta_j$.
Transition $\alpha$ is not a $B^2$-bottleneck, since $u$ (and $q$) do not contain such bottlenecks,
  and all configurations $c$ along $u$ have $c(d) < B$ for all $d \in \Delta_j$ by definition of $r$.
Hence, we must have $c(r_2) > B$ meaning $r_2 \not\in \Delta_j$. 
Exactly one state in $\Delta_j$ decreases its count in transition $\alpha$, 
  but $\alpha$ strictly decreases $\Phi_j$,
  so it must be that both $p_1 \not\in \Delta_j$ and $p_2 \not\in \Delta_j$.
We take $d_j = r_1, s_j = r_2, o_j = p_1$ and $o_j' = p_2$.

There are $|\Lambda_n|^2$ different types of transitions.
As each transition in $u$ decreases $\Phi_j$ by exactly one and there are at least 
  $|\Lambda_n|^2 \cdot b$ such instances, at least one transition type must repeat 
  in $u$ at least $b$ times, completing the proof.
\end{proof}
\begin{claim}
\label{clm:surgery1}
There exist configurations $e : \confspacem$ 
  and $z'$ with $z'_{>0}  \subseteq \Gamma_g$, such that $e + u + x_m \reach z'$.
Moreover, we have an upper bound on the counts of states in $e$:
  $\forall s \in \Lambda_m: e(s) \leq 2^{|\Lambda_m|} \cdot g(m)$.
\end{claim}
\begin{proof}
The proof is analogous to~\cite{DS15}, but we consider a subsequence of the ordered transitions
  $\Delta_b = \{d_1, \ldots, d_k\}$ obtained earlier by~\lemmaref{lem:ordering}.
Since $b(m) \geq g(m)$, we can represent $\Delta_g = \{d_{j_1}, \ldots, d_{j_l}\}$,
  with $j_1 \leq \ldots \leq j_l$.
We iteratively add groups of transitions at the end of transition sequence $p_m$,
  ($p_m$ is the transition sequence from $x_m$ to $y$),
  such that, after the first iteration, the resulting configuration does not contain any agent in $d_{j_1}$.
Next, we add group of transitions and the resulting configuration will not contain any agent agent 
  in $d_{j_1}$ or $d_{j_2}$, and we repeat this $l$ times.
In the end, no agents will be in states from $\Delta_g$.

The transition ordering lemma provides us with the transitions to add.
Initially, there are at most $g(m)$ agents in state $d_{j_1}$ in the system
  (because of the requirement in~\theoremref{thm:surgery} on counts in $u + y$).
So, in the first iteration, we add the same amount (at most $g(m)$)
  of transitions $d_{j_1}, s_{j_1} \to o_{j_1}, o_{j_1}'$, after which, as 
  $s_{j_1}, o_{j_1}, o_{j_1}' \not \in \{d_1, \ldots d_{j_1}\}$,
  the resulting configuration will not contain any agent in configuration $d_{i_1}$.
If there are not enough agents in the system in state $s_{j_1}$ already to add all these transitions,
  then we add the remaining agents in state in $s_{j_1}$ to $e$.
For the first iteration, we may need to add at most $g(m)$ agents.

For the second iteration, we add transitions of type $d_{j_2}, s_{j_2} \to o_{j_2}, o_{j_2}'$ 
  to the resulting transition sequence.
Therefore, the number of agents in $d_{j_2}$ that we may need to consume is at most $3 \cdot g(m)$,
  $g(m)$ of them could have been there in $y + u$, and we may have added $2 \cdot g(m)$ in the 
  previous iteration, if for instance both $o_{j_1}$ and $o_{j_1}'$ were $d_{j_2}$.
In the end, we may need to add $3 \cdot g(m)$ extra agents to $e$.

If we repeat these iterations for all remaining $r=3, \ldots, l$, in the end we will end up in a 
  configuration $z$ that contains all agents in states in $\Gamma_g$ as desired, because of 
  the property of transition ordering lemma that $s_{j_r}, o_{j_r}, o_{j_r}' \not\in \{d_1, \ldots, d_{j_r}\}$.
For any $r$, the maximum total number of agents we may need to add to $e$ 
  at iteration $r$ is $(2^{r} - 1) \cdot g(m)$.
The worst case is when $o_{j_1}$ and $o_{j_1}'$ are both $d_{j_2}$, 
  and $o_{j_2}, o_{j_2}'$ are both $d_{j_3}$, etc.

Finally, it must hold that $l < |\Lambda_m|$, because the final configuration contains $m$ agents 
  in states in $\Gamma_g$ and none in $\{d_{j_1}, \ldots, d_{j_l}\}$, so $\Gamma_g$ cannot be empty.
Therefore, the total number of agents added to $e$ is 
  $g(m) \cdot \sum_{r=1}^{l} (2^r - 1) < 2^{l+1} \cdot g(m) \leq 2^{|\Lambda_m|} \cdot g(m)$.
This completes the proof because $e(s)$ for any state $s$ can be at most the number of agents in $e$,
  which is at most $2^{|\Lambda_m|} \cdot g(m)$.
\end{proof}
\begin{claim}
\label{clm:surgery2}
Let $e$ be any configuration satisfying $\forall s \in \Lambda_m: e(s) \leq 2^{|\Lambda_m|} \cdot g(m)$.
There exist configurations $p$ and $w$, such that 
  $p_{>0} \subseteq \Delta_b$, $w_{>0} \subseteq \Gamma_g$ and $p + x_m \reach p + w + e^{\Delta_g}$.
Moreover, for counts in $p$, we have that $\forall s \in \Lambda_m: p(s) \leq b(m)$
  and for counts in $w^{\Gamma_g}$, we have $\forall s \in \Gamma_g: w(s) \geq 2^{|\Lambda_m|} \cdot g(m)$.
\end{claim} 
\begin{proof}
As in the proof of~\claimref{clm:surgery1}, we define a subsequence ($j_1 \leq j_l$),
  $\Delta_g = \{d_{j_1}, \ldots, d_{j_l}\}$ of $\Delta_b = \{d_1, \ldots, d_k\}$ 
  obtained using~\lemmaref{lem:ordering}.
We start by the transition sequence $p_m$ from configuration $x_m$ to $y$,
  and perform iterations for $r=1, \ldots k$.
At each iteration, we modify the transition sequence, possibly add some agents to configuration $p$,
  which we will define shortly, 
  and consider the counts of all agents not in $p$ in the resulting configuration.
Configuration $p$ acts as a buffer of agents in certain states that we can temporarily borrow.
For example, if we need $5$ agents in a certain state with count $0$ to complete some iteration $r$,
  we will temporarily let the count to $-5$ (add $5$ agents to $p$), and then we will fix the count 
  of the state to its target value, which will also return the ``borrowed'' agents 
  (so $p$ will also appear in the resulting configuration).
As in~\cite{DS15}, this allows us let the counts of certain states temporarily drop below $0$.

We will maintain the following invariants on the count of agents, excluding the agents in $p$,
  in the resulting configuration after iteration $r$:
\begin{itemize}
\item[1)] The counts of all states (not in $p$) in $\Delta_g \cap \{d_1, \ldots, d_r\}$ 
  match to the desired counts in $e^{\Delta_g}$.
\item[2)] The counts of all states in $\{d_1, \ldots d_r\} \setminus \Delta_g$ 
  are at least $2^{|\Lambda_m|} \cdot g(m)$.
\item[3)] The counts in any state diverged by at most 
  $(3^r-1) \cdot 2^{|\Lambda_m|} \cdot g(m)$ from the respective counts in $y$.
\end{itemize}

These invariants guarantee that we get all the desired properties after the last iteration.
Let us consider the final configuration after iteration $k$.
Due to the first invariant, the set of all agents (not in $p$) in states $\Delta_g$ is exactly $e^{\Delta_g}$.
All the remaining agents (also excluding agents in $p$) are in $w$, and thus, 
  by definition, the counts of states in $\Delta_g$ in configuration $w$ will be zero, as desired.
The counts of agents in states $\Delta_b - \Delta_g = \{d_1, \ldots d_k\} - \Delta_g$ that belong to $w$ 
  will be at least $2^{|\Lambda_m|} \cdot g(m)$, due to the second invariant.
Finally, the counts of agents in $\Gamma_b$ that belong to $w$ will also be at least
  $b(m) - 3^{|\Lambda_m|} \cdot 2^{|\Lambda_m|} \cdot g(m) \geq 2^{|\Lambda_m|} \cdot g(m)$,
  due to the third invariant and the fact that the states in $\Gamma_b$ had counts at least $b(m)$ in $y$.
Finally, the third invariant also implies the upper bound on counts in $p$.
The configuration $p$ will only contain the agents in states $\Delta_b$, 
  because the agents in $\Gamma_b$ have large enough starting counts in $y$ borrowing is never necessary.

In iteration $d_r$, we fix the count of state $d_r$.
Let us first consider the case when $d_r$ belongs to $\Delta_g$.
Then, the target count is the count of the state $d_r$ in $e^{\Delta_g}$, 
  which we are given is at most $2^{|\Lambda_m|} \cdot g(m)$.
Combined with the third invariant, the maximum amount of fixing required may be
  is $3^{r-1} \cdot 2^{|\Lambda_m|} \cdot g(m)$.
If we have to reduce the number of $d_r$,
  then we add new transitions $d_r, s_r \to o_r, o_r'$, similar to~\claimref{clm:surgery1}
  (as discussed above, not worrying about the count of $s_r$ possibly turning negative).
However, in the current case, we may want to increase the count of $d_r$.
In this case, we remove instances of transition $d_r, s_r \to o_r, o_r'$ from the transition sequence.
The transition ordering lemma tells us that there are at least $b(m)$ of these transitions
  to start with, so by the third invariant, we will always have enough transitions to remove.
We matched the count of $d_r$ to the count in $e^{\Delta_g}$, so the first invariant still holds.
The second invariant holds as we assumed $d_r \in \Delta_g$
  and since by~\lemmaref{lem:ordering}, $s_r, o_r, o_r' \not\in \{d_1, \ldots, d_r\}$.
The third invariant also holds, because we performed at most $3^{r-1} \cdot 2^{|\Lambda_m|} \cdot g(m)$
  transition additions or removals, each affecting the count of any other given state by at most $2$, 
  and hence the total count differ by at most
$$(3^{r-1}-1) \cdot 2^{|\Lambda_m|} \cdot g(m) + 2 \cdot 3^{r-1} \cdot 2^{|\Lambda_m|} \cdot g(m) = 
 (3^{r}-1) \cdot 2^{|\Lambda_m|} \cdot g(m).$$

Now assume that $d_r$ belongs to $\Delta_b - \Delta_g$.
If the count of $d_r$ is already larger than $2^{|\Lambda_m|} \cdot g(m)$,
  than we do nothing and move to the next iteration, and all the invariants will hold.
If the count is smaller than $2^{|\Lambda_m|} \cdot g(m)$, 
  then we set the target count to $2^{|\Lambda_m|} \cdot g(m)$ and add or remove transitions 
  as in the previous case, and the first two invariants will hold after the iteration.
The only case when the count might require fixing by more than 
  $(3^{r-1}-1) \cdot 2^{|\Lambda_m|} \cdot g(m)$ is when it originally 
  was between $g(m)$ and $2^{|\Lambda_m|} \cdot g(m)$ and decreased.
Then, as in the previous case, the maximum amount of fixing required is at most   
  $3^{r-1} \cdot 2^{|\Lambda_m|} \cdot g(m)$ and considering the maximum effect on counts, 
  the new differences can be at most $3^r \cdot 2^{|\Lambda_m|} \cdot g(m)$.
As before, we also have enough transitions to remove and 
  the third invariant holds.
\end{proof}
\begin{lemma}
\label{lem:stability}
Consider a population protocol in a system with any fixed number of agents $n$, 
  and an arbitrary fixed function $h:\mathbb{N} \to \mathbb{N}^{+}$ such that $h(n) \geq 2^{|\Lambda_n|}$.
Let $\xi(n) = 2^{|\Lambda_n|}$.
For all configurations $c, c' : \confspacen$, such that 
  $c_{>0} \subseteq \Gamma_h(c) \subseteq \Gamma_{\xi}(c')$, 
  any state producible from $c$ is also producible from $c'$. 
Formally, for any state $s \in \Lambda_n$, 
  $c \reach y$ with $y(s) > 0$ implies $c' \reach y'$ with $y'(s) > 0$.
\end{lemma}
\begin{proof}
Since $h(n) \geq 2^{|\Lambda_n|}$, for any state from $\Gamma_h(c)$, its count in $c$ 
  is at least $2^{|\Lambda_n|}$.
As $\Gamma_h(c) \subseteq \Gamma_{\xi}(c')$, the count of each of these states in $c'$ 
  is also at least $\xi(n) = 2^{|\Lambda_n|}$.
We say two agents have the same type if they are in the same state in $c$.
We will prove by induction that any state that can be produced by some transition sequence from $c$, 
  can also be produced by a transition sequence in which at most $2^{|\Lambda_n|}$ agents
  of the same type participate (ever interact).
Configuration $c$ only has agents with types (states) in $\Gamma_h(c)$, 
  and configuration $c'$ also has at least $2^{|\Lambda_n|}$ agents for each of those types, 
  the same transition sequence can be performed from $c'$ to produce the same state as from $c$, 
  proving the desired statement.

The inductive statement is the following. 
There is a $k \leq |\Lambda_n|$, such that for each $i = 0, 1, \ldots, k$ we can find sets 
  $S_0 \subset S_1 \subset \ldots \subset S_k$ where $S_k$ contains all the states that are 
  producible from $c$, and all sets $S_j$ satisfy the following property.
Let $A_j$ be a set consisting of $2^j$ agents of each type in $\Gamma_h(c)$,
  out of all the agents in configuration $c$ (we could also use $c'$), 
  for the total of $2^j \cdot |\Gamma_h(c)|$ agents.
There are enough agents of these types in $c$ (and in $c'$) as $j \leq k \leq |\Lambda_n|$.
Then, for each $0 \leq j \leq k$ and each state $s \in S_j$, 
  there exists a transition sequence from $c$ in which only the agents in $A_j$ ever interact and 
  in the resulting configuration, one of these agents from $A_j$ ends up in state $s$.

We do induction on $j$ and for the base case $j=0$ we take $S_0 = \Gamma_h(c)$.
The set $A_0$ as defined contains one ($2^0$) agent of each type in $\Gamma_h(c) = S_0$\footnote{In $c$, all the agents are in one of the states of $\Gamma_h(c)$, so as long as $n>0$ there must be at least one agent per state (type). So, if $\Gamma_h(c) = \emptyset$, then $n$ must necessarily be $0$, so nothing is producible $A_0 = \emptyset$, $k=0$ and we are done}.
All states in $S_0$ are immediately producible by agents in $A_0$ via an empty transition sequence 
  (without any interactions).

Let us now assume inductive hypothesis for some $j \geq 0$.
If $S_j$ contains all the producible states from configuration $c$, 
  then $k=j$ and we are done.
We will have $k \leq |\Lambda_n|$, because $S_0 \neq \emptyset$ and $S_0 \subset S_1 \subset \ldots S_k$
  imply that $S_k$ contains at least $k$ different states, and there are $|\Lambda_n|$ total.
Otherwise, there must be some state $s \not \in S_j$ that can be produced after an interaction between
  two agents both in states in $S_j$, let us say by a transition 
  $\alpha: r_1, r_2 \to s, p$ with $r_1, r_2 \in S_j$ (or there is no state that cannot already be produced).
Also, as $S_j$ contains at least $j$ states out of $|\Lambda_n|$ total, 
  and there is the state $s \not \in S_j$, $j < |\Lambda_n|$ holds and the set $A_{j+1}$ is well-defined. 
Let us partition $A_{j+1}$ into two disjoint sets $B_1$ and $B_2$ where each contain 
  $2^j$ agents from $c$ for each type.
Then, by induction hypothesis, there exists a transition sequence where only the agents in $B_1$ 
  ever interact and in the end, one of the agents $b_1 \in B_1$ ends up in the state $r_1$.
Analogously, there is a transition sequence for agents in $B_2$, after which an agent $b_2 \in B_2$
  ends up in state $r_2$.
Combining these two transition and adding one instance of transition $\alpha$ in the end 
  between agents $b_1$ and $b_2$ (in states $r_1$ and $r_2$ respectively) leads to a configuration
  where one of the agents from $A_{j+1}$ is in state $s$.
Also, all the transitions are between agents in $A_{j+1}$.
Hence, setting $S_{j+1} = S_j \cup \{s\}$ completes the inductive step.
\end{proof}
\begin{corollary}
\label{crl:boundmaj}
Any monotonic population protocol with $|\Lambda_n| \leq 1/2 \log{\log{n}}$ states 
  for all sufficiently large number of agents $n$ that stably computes correct majority decision 
  for initial configurations with majority advantage $\epsilon n$, must take 
  $\Omega \left( \frac{n}{36^{|\Lambda_n|} \cdot |\Lambda_n|^6 \cdot \max(2^{|\Lambda_n|},\epsilon n)^2} \right)$ 
  expected parallel time to stabilize.
\end{corollary}
\begin{proof}
We set $g(n) = \max(2^{|\Lambda_m|+1}, 4 \epsilon n)$.
For majority computation, initial configurations consist of agents in one of two states,
  with the majority state holding an $\epsilon n$ advantage in the counts.
Therefore, the sum of two initial configurations of the same protocol is also a valid initial configuration, 
  and thus monotonic populations protocols for majority computation must be input-additive.
The bound is nontrivial only in a regime $\epsilon n \in o(\sqrt{n})$, 
  which we will henceforth assume without loss of generality.
The initial configurations we consider from $I_n$ will all have advantage $\epsilon n$,
  and are all be fully dense.

Let us prove that for all sufficiently large $m$, in any final stable configuration $y$, 
  strictly less than $2^{|\Lambda_m|} \leq g(m) / 2$ agents will be in the initial minority state $B_m$.
The reason is that if $c$ is the initial configuration of all $m$ agents in state $B_m$, the protocol 
  must stabilize from $c$ to a final configuration where the states correspond to decision $\id{Win}_B$.
By~\lemmaref{lem:stability}, from any configuration that contains at least $2^{|\Lambda_m|}$ agents in $B_m$
  it would also be possible to reach a configuration where some agent supports decision $B_m$.
Therefore, all stable final configuration $y$ have at most $g(m)/2 - 1$ agents in initial minority state $B_m$.
This allows us to let $u$ be a configuration of $g(m)/2 + 1 \geq 2 \epsilon m + 1$ agents in state $B_m$.

Assume, to the contrary, that the protocol stabilizes in parallel time 
  $o \left( \frac{n}{36^{|\Lambda_n|} \cdot |\Lambda_n|^6 \cdot \max(2^{|\Lambda_n|}, \epsilon n)^2} \right)$.
We only consider initial configurations that are fully dense and contain $\frac{(1+\epsilon)n}{2}$ agents 
  in state $A_n$ and $\frac{(1-\epsilon)n}{2}$ agents in state $B_n$, i.e. 
  having majority state $A_n$ with advantage $\epsilon n$.
Let $I'_n \subseteq I_n$ contain only this fully dense initial configuration for each $n$ and  
  using~\theoremref{thm:surgery} with $I'_n$ instead of $I_n$, 
  we can find infinitely many configurations $i$ and $z$ of at most $3m$ agents, such that 
  (1) $i + u \reach z$,
  (2) $i \in I'_{2m}$, i.e. it is an initial configuration of $2m$ agents 
  with majority state $A_{2m} = A_m$ and advantage $2 \epsilon m$.
  (3) $|\Lambda_m| = |\Lambda_{2m}| = |\Lambda_{2m + |u|}| = |\Lambda_{3m}|$ and by monotonicity,
  the same protocol is used for all number of agents between $m$ and $3m$, 
  (4) $z_{>0} \subseteq \Gamma_g(y)$, i.e all agents in $z$ are in states 
  that have counts at least $g(m)$ in some stable output configuration $y$ of $m$ agents.

To get the desired contradiction we will prove two things.
First, $z$ is actually a stable output configuration for decision $\id{Win}_A$ (majority opinion in $i$),
  and second, $i + u$ is a valid initial configuration for the majority problem,
  but with majority state $B_{m} = B_{2m+|u|}$.
This will imply that the protocol stabilize to a wrong outcome, and complete the proof by contradiction.

If we could reach a configuration from $z$ with any agent in a state $s$ that maps to output $\id{Win}_B$
  ($\gamma_m(s) = \id{Win}_B$), then by~\lemmaref{lem:stability}, from a configuration $y$ (which contains 
  $2^{|\Lambda_m|}$ agents in each of the states in $\Gamma_g(y)$) we can also reach a 
  configuration with an agent in a state $s$ that maps to output $\id{Win}_B$.
However, configuration $y$ is a final stable configuration for an initial 
  configuration in $I_m$ with a majority $A_m$.

Configuration $i \in I_{2m}$ contains $2 \epsilon m$ more agents in state $A_m$ states than in state $B_m$.
Configuration $u$ consists of at least $2 \epsilon m + 1$ agents all in state $B_m$.
Hence, $i + u$ which is a legal initial configuration from $I_{2m + |u|}$ has a majority of agents in state $B_m$.
\end{proof}

\section{Analysis of the Majority Algorithm}
The update rules in~\figureref{fig:logmajo} are chained, 
  i.e. a cancel is followed by a join and a split.
This is an optimization, applying as many possible reactions as possible.
However, for the analysis we consider a slight modification,
  where we only apply split only if both join and cancel were unsuccessful. 

For presentation purposes, we assume that $n$ is a power of two,
  and when necessary, we assume that it is sufficiently large.
Throughout this proof, we denote the set of nodes executing the protocol by $V$. 
We measure execution time in discrete steps (rounds), 
  where each time step $t$ corresponds to an interaction. 
The \emph{configuration} at a given time $t$ is a function $c : V \rightarrow Q$, 
  where $c(v)$ is the state of the node $v$ at time $t$. 
  (We omit the explicit time $t$ when clear from the context.)
  
Recall that a value of a state $\langle x, y \rangle$ is defined as $x-y$ and 
  we will also refer to $\max(x, y)$ as the \emph{level} of this node.  
We call $\langle x, y \rangle$ a \emph{mixed} state if both $x$ and $y$ are non-zero, 
  and a \emph{pure} state otherwise. 
A mixed or pure node is a node in a mixed or a pure state, respectively.

The rest of this section is focused on proving the following result.
\begin{theorem}
\label{thm:majcorrect}
The Split-Join algorithm will never stabilize to the minority decision, 
  and is guaranteed to stabilize to the majority decision within $O( \log^3 n )$ parallel time, 
  both in expectation and w.h.p. 
\end{theorem}

\paragraph{Correctness}
We first prove that nodes never stabilize to the sign of the initial minority (safety), 
  and that they eventually stabilize to the sign of the initial majority (termination).

The first statement follows since given the interaction rules of the algorithm, 
  the sum of the encoded values stays constant as the algorithm progresses.
The proof follows by the structure of the algorithm.
\begin{invariant}
\label{sum-invariant}
The sum $\sum_{v \in V} \id{value}(c(v))$ never changes, 
  for all reachable configurations $c$ of the protocol.
\end{invariant}
\noindent This invariant implies that the algorithm may never stabilize to a wrong decision value.
For instance, if the initial sum is positive, then positive values must always exist in the system. 
Therefore we only need to show that the algorithm stabilizes to configurations where 
  all nodes have the same sign.
We do this via a rough bound, assuming an arbitrary starting configuration.
\begin{claim}
\label{clm:splitbd}
There are at most $2n^2$ split reactions in any execution.
\end{claim}
\begin{proof}
A level of a node in state $s = \langle x, y \rangle$ is defined as $\id{level}(s) = \max(x, y)$.
Consider a node in a state with level $l$.
Then, we say that the \emph{potential of the node} is $\phi(l) = 2l$ for $l > 0$ and $\phi(0) = 1$.
In any configuration $c$, the \emph{potential of the system} is $\Phi(c) = \sum_{i=1}^n \phi(\id{level}(s_i))$.

Then, the potential of the system in the initial configuration is $\sum (2n) = 2 n^2$, 
  and it can never fall below $\sum (1) = n$.
By the interaction rules of the algorithm, potential of the system never increases after an interaction,
  and it decreases by at least one after each successful \lit{split} interaction.
This implies the claim.
\end{proof} 
\begin{lemma}
\label{lem:lowp}
Let $c$ be an arbitrary starting configuration. 
Define $S := \sum_{v \in V} \id{value}(c(v)) \neq 0$.
With probability $1$, the algorithm will reach a configuration $\hat{c}$ 
  such that $\id{sgn}(\hat{c}(v)) = \sgn(S)$ for all nodes $v \in V$.
Moreover, in all later configurations $c_e$ reachable from $\hat{c}$, no node can ever have a different sign, 
  i.e. $\forall v \in V: \id{sgn}(c_e(v)) = \id{sgn}(\hat{c}(v))$.
For sufficiently large $n$, the stabilization time to $\hat{c}$ is at most $n^5$ expected communication rounds, 
  i.e. parallel time $n^4$.
\end{lemma}
\begin{proof}
Assume without loss of generality that the sum $S$ is positive. 

We estimate the expected stabilization time by splitting the execution into three phases.
The first phase starts at the beginning of the execution, and lasts until either 
  i) no node encodes a strictly negative value \emph{or} 
  ii) each node encodes a value in $\{-1, 0, 1\}$, i.e. all nodes are in states 
  $\langle 1, 0 \rangle$, $\langle 0,0 \rangle^{+}$. $\langle 0,0 \rangle^{-}$ or $\langle 0, 1 \rangle$.

Due to~\invariantref{sum-invariant}, at least one node encodes a strictly positive value. 
Also, by definition, during the first phase there is always a node encoding a strictly negative value.
Moreover, there is a node in state $\langle x,y \rangle$ with $\max(x, y) > 1$.
Assume that $x > y$ for this node.
Then, if there is another node in state $\langle 0, y_2 \rangle$ for any $y_2$, 
  then with probability at least $1/n^2$ these two nodes interact in the next round resulting in a split reaction.
Otherwise, every node $\langle x_1, y_1 \rangle$ that encodes a strictly negative value 
  must have $\min(x_1,y_1) > 0$.
At least one such node exists and if there is another node in state $\langle x_2, 0 \rangle$ for any $x_2$, 
  then again with probability at least $1/n^2$ a split reaction occurs in the next round. 
The case of $x < y$ is analogous and we get that during the first phase, 
  if there is no pair whose interaction would result in a split reaction, 
  all nodes must be in states $\langle x,y \rangle$ with $\min(x,y) > 0$, i.e. in mixed states.
By~\claimref{clm:cauchy}, with probability at least $\frac{1}{2(\log{n}-1)}$ a pure node appears after 
  the next communication round and by the above argument, if the first phase has not been completed,
  in the subsequent round a split reaction will occur with probability at least $1/n^2$.
Therefore, during the first phase, the expected number of rounds until the next split reaction is 
  at most $4 n^2 (\log{n}-1)$.
By~\claimref{clm:splitbd}, there can be at most $2n^2$ split reactions in any execution, thus
  the expected number of communication rounds in the first phase is at most $8 n^4 (\log{n}-1)$.

The second phase starts immediately after the first, and ends when no node encodes a strictly negative value.
Note that if this was already true when the first phase ended, then the second phase is trivially empty. 
Consider the other case when all nodes encode values $-1$, $0$ and $1$ at the beginning of the second phase.
Under these circumstances, because of the update rules, 
  no node will ever be in a state $\langle x,y \rangle$ with $\max(x,y) > 1$ in any future configuration.
Also, the number of nodes encoding non-zero values can only decrease.  
In each round, with probability at least $1/n^2$, two nodes with values $1$ and $-1$ interact,
  becoming $\langle 0,0 \rangle^+$ and $\langle 0,0 \rangle^-$.
Since this can only happen $n/2$ times, 
  the expected number of communication rounds in the second phase is at most $n^3/2$.

The third phase lasts until the system stabilizes, 
  that is, until all nodes with value $0$ are in state $\langle 0,0 \rangle^{+}$.
By~\invariantref{sum-invariant}, $S > 0$ holds throughout the execution, 
  so there is at least one node with a positive sign and non-zero value. 
There are also at most $n-1$ \emph{conflicting} nodes with negative sign, all in state $\langle 0,0 \rangle^{-}$.
Thus, independently in each round, with probability at least $1/n^2$, 
  a conflicting node meets a node with strictly positive value and becomes $\langle 0,0 \rangle^{+}$,
  decreasing the number of conflicting nodes by one.
The number of conflicting nodes can never increase and when it becomes zero, 
  the system has stabilized to the desired configuration $\hat{c}$.
Therefore, the expected number of rounds in the third phase is at most $n^3$.

Combining the results and using the linearity of expectation,
  the total expected number of communication rounds before reaching $\hat{c}$ is 
  at most $n^3(8n (\log{n}-1) + 1/2 + 1) \leq n^5$ for sufficiently large $n$.
Finite expectation implies that the algorithm stabilizes with probability $1$.
Finally, when two nodes with positive sign meet, they both remain positive, 
  so any configuration $c_e$ reachable from $\hat{c}$ has the correct signs.
\end{proof}

\paragraph{Stabilization Time}
Next, we bound the time until all nodes stabilize to the correct sign.
\begin{claim}
\label{clm:cauchy}
Consider a configuration where $n \delta$ out of the $n$ nodes are in a mixed state, 
  for $\delta \geq \frac{2 (\log{n}-1)}{n}$. 
In the next interaction round, the number of mixed nodes strictly decreases 
  with probability at least $\frac{\delta^2}{2(\log{n} - 1)}$.   
\end{claim}
\begin{proof}
Consider $\log{n} - 1$ buckets corresponding to values $1, 2, 4, \ldots, n/4$.
Let us assign mixed nodes to these buckets according to their states,
  where node in state $\langle x,y \rangle$ goes into bucket $\min(x, y)$.
All nodes fall into one of the $\log{n}-1$ buckets because of the definition of (mixed) states.

If two nodes in the same bucket interact, either cancel or join will be successful, 
  and since we consider the algorithm where split is not applied in this case, 
  and the number of mixed nodes will strictly decrease.
Thus, if there are $d_1, d_2, \ldots, d_{\log{n}-1}$ nodes in the buckets, 
  the number of possible interactions that decrease the number of mixed nodes 
  is at least $\sum_{i=1}^{\log{n}-1} \frac{d_i(d_i-1)}{2} = \frac{(\sum d_i^2) - n\delta}{2}$.    

By the Cauchy-Schwartz inequality, $\sum d_i^2 \geq \frac{n^2 \delta^2}{\log{n}-1}$.
Combining this with the above and using $n\delta \geq 2(\log{n}-1)$ we get that the 
  there are at least $\frac{n^2\delta^2}{4(\log{n}-1)}$ pairs of nodes whose interactions 
  decrease the number of mixed nodes.
The total number of pairs is $n(n-1)/2$, proving the desired probability bound.  
\end{proof}
\begin{claim}
\label{clm:randomwalk}
Suppose $f$ is a function such that $f(n) \in O(\poly(n))$.
For all sufficiently large $n$, 
  the probability of having less than $\frac{n}{2^{11} \log{n}}$ pure nodes in the system 
  at any time during the first $f(n)$ communication rounds is at most $1 - 1/n^5$.
\end{claim}
\begin{proof}
Assume that this number became less than $\frac{n}{2^{11} \log{n}}$ 
  for the first time at time $T$ after some number of communication rounds.
Let $t$ be the last time when the number of pure nodes was at least $\frac{n}{2^9 \log{n}}$ 
  (such a time exists since the initial number of pure nodes is $n$)
  and let $\alpha$ be the number of communication rounds between $t$ and $T$.
The number of mixed nodes increases by at most two in each round, so $\alpha \geq \frac{n}{2^{11}\log{n}}$.

By definition of $t$ and $T$, at all times during the $\alpha$ communication rounds between $t$ and $T$,
  at least $\frac{n (2^9\log{n} - 1)}{2^9 \log{n}} \geq \frac{n}{2}$ nodes are mixed.
Thus, by~\claimref{clm:cauchy} in each of these communication rounds, 
  the number of mixed nodes decreases by at least one with probability at least $\frac{1}{8\log{n}}$.
Let us describe by a random variable $X \sim \mathrm{Bin}(\alpha, \frac{1}{8\log{n}})$
  at least how often the number of mixed nodes decreased.
Each node is pure or mixed, and by Chernoff Bound, the probability that the number of 
  pure nodes increased less than $\frac{\alpha}{16 \log{n}}$ times is
  $\Pr\left[X \leq \frac{\alpha}{16 \log{n}} \right] = \Pr\left[X \leq \frac{\alpha}{8 \log{n}} \left( 1 - 1/2 \right)\right] \leq \exp\left(- \frac{\alpha}{8 \log{n} \cdot 2^2 \cdot 2} \right) \leq \exp\left(- \frac{n}{2^{17} \log^2{n}} \right)$

On the other hand, in each of these $\alpha$ rounds, 
  the number of pure nodes can decrease only if one of the interacting nodes was in a pure state.
By definition of $t$ and $T$, the number of such pairs is at most 
  $\frac{n^2}{2^{18}\log^2{n}} + 2 \frac{n^2 (2^9 \log{n} - 1)}{2^{18} \log^2{n}} \leq \frac{2 n^2}{2^9 \log{n}}$.
This implies that in each round the probability that the number of pure nodes will decrease is 
  at most $\frac{1}{2^6 \log{n}}$.
Let us describe the (upper bound on the) number of such rounds by a random variable 
  $Y \sim \mathrm{Bin}(\alpha, \frac{1}{2^6\log{n}})$.
Since in each such round the number of pure nodes can decrease by at most $2$,
  using Chernoff bound the probability that the number of pure nodes decreases 
  by more than $\frac{\alpha}{16 \log{n}}$ during the $\alpha$ communication rounds is at most 
$\Pr\left[Y \geq \alpha/(32 \log{n}) \right] \leq \exp \left( -\frac{n}{2^{19} \log^2{n}}\right)$

In order for the number of pure nodes to have decreased 
  from $\frac{n}{2^9 \log{n}}$ at time $t$ to $\frac{n}{2^{11} \log{n}}$ at time $T$,
  either the number of mixed nodes must have increased by at most $\frac{\alpha}{16 \log{n}}$, 
  or the number of pure nodes must have decreased by at least $\frac{\alpha}{16 \log{n}}$
  during the $\alpha$ communication rounds between $t$ and $T$.
Otherwise, the increase in mixed nodes would be more than the decrease in pure nodes.
However, by union bound, the probability of this is at most
  $\exp\left(- \frac{n}{2^{17} \log^2{n}} \right) + \exp \left( -\frac{n}{2^{19} \log^2{n}}\right)$.

We can now take union bound over the number of communication rounds 
  until the number of pure nodes drops below $\frac{n}{2^{11} \log{n}}$ (time $T$).
For at most $f(n) \in O(\poly(n))$ rounds, we get that the probability of the number of pure nodes 
  ever being less than $\frac{n}{2^{11} \log{n}}$ is at most $1/n^5$ for all large enough $n$.
\end{proof}
Consider the high probability case of the above claim,
  where a fraction of pure nodes are present in every configuration in the execution prefix.
We call a round a \emph{negative-round} if, in the configuration $c$ at the beginning of the round,
  there are at least $\frac{n}{2^{11} \log{n}}$ pure nodes and 
  at least \emph{half} of the pure nodes encode a non-positive value.
Analogously, we call a round a \emph{positive-round} 
  if there are at least $\frac{n}{2^{11} \log{n}}$ pure nodes,
  at least half of which encode a non-negative value.
A round can be simultaneously negative \emph{and} positive, 
  for instance when all pure nodes encode value $0$. 
Next claim establishes the speed at which the maximum level in the system decreases.
The proof follows by bounding the probability that a node with the maximum level meets a pure node 
  with value $0$ or a value of the opposite sign.
This results in a split (or cancel) reaction decreasing the level of the node, 
  and we use Chernoff and Union Bounds to bound the probability that the node 
  avoids such a meeting for significant time. 
\begin{claim}
\label{clm:elimchf}
Let $w > 1$ be the maximum level among the nodes with a negative (resp., positive) sign.
There is a constant $\beta$, such that after $\beta n \log^2{n}$ positive-rounds (resp., negative-rounds) 
  the maximum level among the nodes with a negative (resp., positive) sign 
  will be at most $\lfloor w/2 \rfloor$ with probability at least $1-\frac{1}{n^5}$.
\end{claim}
\begin{proof}
We will prove the claim for nodes with negative values. (The converse claim follows analogously.) 
Fix a round $r$, and recall that $w > 1$ is the maximum level 
  of a node with a negative value at the beginning of the round. 
Let $U$ be the set of all nodes with negative values and the same level $w$ at the beginning of the round, 
  and let $u = |U|$. 
We call these nodes \emph{target nodes}.

By the structure of the algorithm, the number of target nodes never increases 
  and decreases by one in every \emph{eliminating} round where
  a target node meets a pure node with a non-negative value, due to a split or cancel reaction.
Consider a set of $\alpha n \log{n}$ consecutive positive-rounds after $r$, for some constant $\alpha > 2^{12}$. 
In each round, if there are still at least $\lceil u/2 \rceil$ target nodes, 
  then the probability of this round being eliminating is at least $\frac{\lceil u/2 \rceil}{2^{12} n \log{n}}$ 
  (since in a positive round at least half of $\frac{n}{2^{11} \log{n}}$ pure nodes have non-negative value).
Let us describe the process by considering a random variable 
  $Z \sim \mathrm{Bin}(\alpha n \log{n}, \frac{\lceil u/2 \rceil}{2^{12} n \log{n}})$, 
  where each success event corresponds to an eliminating round.
By a Chernoff Bound, the probability of having $\alpha n \log{n}$ iterations with at most 
  $\lceil u/2 \rceil$ eliminations is at most:
\begin{align*}
\Pr\left[Z \leq \lceil u/2 \rceil \right] = 
\Pr\left[Z \leq \frac{\alpha \lceil u/2 \rceil}{2^{12}} \left( 1 - \frac{\alpha - 2^{12}}{\alpha}\right)\right]
\leq \exp\left(- \frac{\alpha \lceil u/2 \rceil (\alpha - 2^{12})^2}{2^{13} \alpha^2} \right)
\end{align*}
For sufficiently large $\alpha$ and $u \geq \log{n}$, 
  the probability of this event is at most $\frac{1}{n^6}$ for $\alpha n \log{n}$ positive-rounds.
Applying the same rationale iteratively as long as $u \geq \log n$, we obtain by using a Union Bound that 
  the number of target nodes will become less than $\log n$ 
  within $\alpha n \log{n} (\log n - \log \log n)$ positive-rounds, 
  with probability at least $1 - \frac{\log{n} - \log\log{n}}{n^6}$.

Finally, we wish to upper bound the remaining number of positive-rounds until no target node remains. 
Again for sufficiently large $\alpha$, but when $u < \log{n}$, 
  we get from the same argument as above that the number of target nodes 
  is reduced to $\lfloor u/2 \rfloor$ within $\frac{\alpha n \log^2{n}}{u}$ consecutive positive-rounds 
  with probability $1/n^6$.
So we consider increasing numbers of consecutive positive-rounds, 
  and obtain that no target nodes will be left after at most 
  $\alpha n \log{n} + 2 \alpha n \log{n} + \ldots + \alpha n \log^2{n} \leq 2 \alpha n \log^2{n}$ positive-rounds, 
  with probability at least $1 - \frac{\log\log{n}}{n^6}$, where we have taken the union bound over $\log \log n$ events. 
The original claim follows by setting $\beta = 3 \alpha$, 
  taking Union Bound over the above two events ($u \geq \log{n}$ and $u < \log{n}$) and $\log{n} \leq n$.
\end{proof}

We get a condition for halving the maximum level (among positive or negative values) 
  in the system with high probability.
The initial levels in the system is $n$, which can only be halved $\log{n}$ times for each sign.
Combining everything results in the following claim:
\begin{claim}
\label{clm:thrcases}
There exists a constant $\beta$, such that if during the first $2 \beta n \log^3{n}$ rounds 
  the number of pure nodes is always at least $\frac{n}{2^{11} \log{n}}$,  
  then with probability at least $1-\frac{2\log{n}}{n^5}$,
  one of the following three events occurs at some point during these rounds: 
  \begin{enumerate}
    \item Nodes only encode values in $\{-1,0,1\}$;
    \item There are less than $\frac{n}{2^{12} \log{n}}$ nodes with non-positive values, all encoding $0$ or $-1$, 
    \item There are less than $\frac{n}{2^{12} \log{n}}$ nodes with non-negative values, all encoding $0$ or $1$.  
  \end{enumerate}
\end{claim}
\begin{proof}
We take a constant $\beta$ that works for~\claimref{clm:elimchf}. 
Since there are at least $\frac{n}{2^{11} \log{n}}$ pure node at all times during the first $2 \beta n \log^3{n}$ rounds, 
  each round during this interval is a negative-round, a positive-round, or both.
We call maximum positive (resp. negative) level the maximum level 
  among all the nodes encoding non-negative (resp. non-positive) values. 
Unless the maximum positive level in the system is $\leq 1$, by~\claimref{clm:elimchf}, a stretch of 
  $\beta n \log^2{n}$ negative-rounds halves the maximum positive level, with probability at least $1-\frac{1}{n^5}$. 
The same holds for stretches of $\beta n \log^2{n}$ positive-rounds and the maximum negative level. 

Assume that none of the three events hold at any time during the first $2 \beta n \log^3{n}$ rounds. 
In that case, each round can be classified as either:
\begin{itemize}
\item a negative-round where the maximum positive level is strictly larger than $1$, or 
\item a positive-round where the maximum negative level is strictly larger than $1$.
\end{itemize} 
To show this, without a loss of generality consider any positive-round 
  (we showed earlier that each round is positive-round or a negative-round).
If the maximum negative level is $>1$ then the round can be classified as claimed,
  thus all non-positive values in the system must be $0$ or $-1$.  
Now if there are less than $\frac{n}{2^{12} \log{n}}$ such nodes, then we have the second event,
  so there must be more than $\frac{n}{2^{12} \log{n}}$ nodes encoding $0$ and $-1$.
However, all these nodes are pure, so the round is simultaneously a negative-round.
Now if the maximum positive level is $>1$ then the round can again be classified as claimed,
  and if the maximum positive level is at most $1$, then all nodes in the system are encoding values $-1$, $0$ or $1$
  and we have the first event. 

Thus, each round contributes to at least one of the stretches of $\beta n \log^2{n}$ rounds 
  that halve the maximum (positive or negative) level, w.h.p.
However, this may happen at most $2\log{n}$ times. 
By applying~\claimref{clm:elimchf} $2\log{n}$ times and the Union Bound
  we get that after the first $2 \beta n \log^3{n}$ rounds, 
  with probability at least $1-\frac{2\log{n}}{n^5}$ only values $-1$, $0$ and $1$ may remain.
However, this is the same as the first event above. 
Hence, the probability that none of these events happen is at most $\frac{2\log{n}}{n^5}$.
\end{proof}
\paragraph{Final Argument} 
To see how this claim can be used to obtain the stabilization upper bound,
  let us assume without loss of generality that the initial majority of nodes was 
  in $A$ (positive) state, i.e. $a > b$.

Setting $\beta$ as in~\claimref{clm:thrcases}, by~\claimref{clm:randomwalk}, with high probability,
  we have at least $\frac{n}{2^{11} \log{n}}$ pure nodes during the first $2 \beta n \log^3{n}$ rounds.
Thus, w.h.p. during these rounds the execution reaches a configuration 
  where one of the three events from~\claimref{clm:thrcases} holds.
Consider this point $T$ in the execution.

By our assumption about the initial majority and~\invariantref{sum-invariant}, 
  $\sum_{v \in V} \id{value}(c(v)) = \epsilon n^2$ holds in every reachable configuration $c$. 
The third event is impossible, because the total sum would be negative. 
In the first event, the total sum is $\epsilon n^2 \geq n$ of $n$ encoded values each being $-1$, 
  $0$ or $1$.
Therefore, in this case, all nodes must be in state $\langle 1, 0 \rangle$ and we are done. 

In the second event implies there are at least 
  $\frac{n (2^{12}\log{n} -1)}{2^{12}\log{n}} \geq \frac{2n}{3}$ nodes encoding strictly positive values.
Hence, at time $T$ during the first $2 \beta n \log^3{n}$ rounds 
  there are at least $n /3$ more strictly positive than strictly negative values. 
Moreover, $-1$'s are the only strictly negative values of the nodes at point $T$, 
  and this will be the case for the rest of the execution because of the update rules.
After time $T$, we have
\begin{repclaim}{clm:4verge}
Consider a configuration where all nodes with strictly negative values encode $-1$,
  while at least $\frac{2n}{3}$ nodes encode strictly positive values.
The number of rounds until stabilization is 
  $O(n \log{n})$ in expectation and
  $O(n \log^2{n})$ with high probability.
\end{repclaim}
Using this, and by Union Bound over~\claimref{clm:randomwalk} and ~\claimref{clm:thrcases},
  with probability $1-\frac{\log{n} + 1}{n^5}$ the number of communication rounds to stabilization 
  is thus $2 \beta n \log^3{n} + O(n \log^2{n}) = O(n \log^3{n})$.

In the remaining low probability event, with probability at most $\frac{\log{n} + 1}{n^5}$, 
  the remaining number of rounds is at most $O(n^5)$ by~\lemmaref{lem:lowp}.
Therefore, the same $O(n \log^3{n})$ bound also holds in expectation, 

\begin{claim}
\label{clm:4verge}
Consider a configuration where all nodes with strictly negative values encode $-1$,
  while at least $\frac{2n}{3}$ nodes encode strictly positive values.
The number of rounds until stabilization is 
  $O(n \log{n})$ in expectation and
  $O(n \log^2{n})$ with high probability.
\end{claim}
\begin{proof}
In any configuration, let us call \emph{conflicting} any node that encodes $-1$, 
  and \emph{target} node any node that has a strictly positive value. 
Because of the structure of the algorithm, 
  and that in configuration $c$ the only nodes with non-positive sign encode $-1$ or $0$,
  in all configurations reachable from $c$ nodes with negative values will also only encode $-1$ or $-0$.
Moreover, the number of conflicting nodes can never increase after an interaction.
Observe that the number of conflicting nodes decreases by one after an interaction where a target node 
  (with a stritly positive value) meets a node with value $-1$,
  while the number of target nodes may also decrease by at most $1$. 
This is because a split reaction happens on the positive component of the target node 
  (since the positive component of the conflicting node is $0$) 
  and both nodes get value $\geq 0$ after the interaction. 
There are at least $n/3$ more target nodes than conflicting nodes in $c$,
  therefore, in every later configuration, there must always be at least $n /3$ target nodes. 

Let us estimate the number of rounds until each conflicting node has interacted with a target node, 
  at which point no more conflicting nodes may exist.
Let us say there were $x$ conflicting nodes in configuration $c$.
The expected number of rounds until the first conflicting node meets a target node is at most $\frac{3n}{x}$,
  since the probability of such an interaction happening in each round is at least $\frac{x}{n} \cdot \frac{n }{3n}$. 
The expected number of rounds for the second node is then $\frac{3n}{(x-1)}$, and so on.
By linearity of expectation, the expected number of rounds until all conflicting nodes are eliminated is 
  $O(n \log x) \leq O(n \log n)$.

At this point, all nodes that do not have a positive sign must be in state $\langle 0, 0 \rangle^-$.
If we redefine \emph{conflicting} to describe these nodes, it is still true that an interaction of a conflicting node 
  with a target node brings the conflicting node to state $\langle 0, 0 \rangle^+$,
  decreasing the number of conflicting nodes.
As we discussed at least $n /3$ target nodes are still permanently present in the system.
By the structure of the algorithm no interaction can increase the number of conflicting nodes,
  and the system stabilizes when all conflicting nodes are eliminated.
This takes expected $O(n \log n)$ rounds by exactly the same argument as above.

To get the high probability claim, simply observe that when there are $x$ conflicting nodes in the system,
  a conflicting node will interact with a target node within $\frac{3n O(\log{n})}{x}$ rounds, with high probability.
The same applies for the next conflicting node, etc.
Taking Union Bound over these events gives the desired result.
\end{proof}

\section{Synthetic Coins}

\begin{claim}\label{clm:coinexp}
$\EX[X_{i+m} \mid X_i = x] = n / 2  + ( 1- 4 / n )^m \cdot ( x - n / 2)$.
\end{claim}
\begin{proof}
If two agents both with coin values one are selected, the number of ones decreases by two.
If both coin values are zero, it increases by two, and otherwise stays the same. 
Hence, we have that
\begin{align*}
  \EX[X_{i+m} \mid X_{i+m-1} = t] &= (t-2) \cdot \Pr[X_{i+m} = t-2] + t \cdot Pr[X_{i+m} = t]
    + (t+2) \cdot \Pr[X_{i+m} = t+2] \\
  &= (t-2) \cdot \frac{t(t-1)}{n(n-1)} + t \cdot \frac{2t(n-t)}{n(n-1)} 
    + (t+2) \cdot \frac{(n-t)(n-t-1)}{n(n-1)} \\
  &= t + \frac{2}{n(n-1)} \cdot \left( n^2 -2nt -n +2t \right) = t 
    \cdot \left( 1-\frac{4}{n} \right) + 2
\end{align*}
Thus, we get a recursive dependence $\EX[X_{i+m}] = \EX[X_{i+m-1}] \cdot (1-4/n) + 2$, 
  that gives
\begin{equation*}
  \EX[X_{i+m}] = 2 \cdot \sum_{j=0}^{m-1} \left( 1-\frac{4}{n} \right)^j
    + \EX[X_{i}] \cdot \left( 1-\frac{4}{n} \right)^m
  = \frac{n}{2} + \left( 1-\frac{4}{n} \right)^m \left( x - \frac{n}{2} \right)
\end{equation*}
\noindent by telescoping.
\end{proof}

\section{Analysis of the Leader Election Algorithm}
\begin{lemma}
\label{lem:lecorrect}
All nodes can never be minions. 
A configuration with $n-1$ minions must have a stable leader, meaning that
  the non-minion node will never become a minion, while minions will remain minions.
\end{lemma}
\begin{proof}
Assume for contradiction that all nodes are minions at some time $T$, 
  and let $u$ be a maximum $(.\id{payoff}, .\id{level})$ pair 
  (lexicographically) among all the minions at this time.
No node in the system could ever have had a larger pair,
  because no interaction can decrease a pair.
The minions only record the values of such pairs they encounter, 
  and never increase them, so there must have been a contender in the system 
  with a payoff and level pair $u$ that turned minion by time $T$.
Among all such contenders, consider the one that turned minion the last.
It could not have interacted with a minion, because no minion (and no node)
  in the system ever held a larger pair.
On the other hand, even if it interacted with another contender, 
  the contender also could not have held a larger pair.
Thus, it could only have been an interaction with another contender, 
  that held the same pair and a larger coin value used as a tie-breaker.
However, that interaction partner would remain a contender and
  and must have turned minion later, contradicting our assumption that 
  the interaction we considered was the last one where a contender with 
  a pair $u$ got eliminated.

By the structure of the algorithm, minions can never change their mode.
In any configuration with $n-1$ minions, the only non-minion must remain
  so forever, and thus be a stable leader, because otherwise we would get $n$
  minions and violate the above argument.
\end{proof}
\begin{lemma}
\label{lem:stage1}
With probability $1 - O(1)/n^3$, after $O(n \log{n})$ interactions, 
  all agents will be in the competition stage,
  that is, either in $\lit{tournament}$ or $\lit{minion}$ mode, 
  with maximum $\id{payoff}$ at least $\log{n} / 2$ and at most $9 \log{n}$,
  and at most $5 \log{n}$ agents will have this maximum $\id{payoff}$.  
\end{lemma}
\begin{proof}
Let us call the first $2n$ interactions the \emph{irrelevant} interactions,
  and all interactions after that the \emph{relevant} interactions.
By~\theoremref{thm:coin}, during relevant interactions, 
  with probability at least $1 - 2\exp(-\sqrt{n}/4)$, there are 
  at least $n (1/2 - 1/2^8)$ and at most $n (1/2 + 1/2^8)$ agents holding each possible coin value.
Taking an union bound, this holds for all of the first $n^4$ relevant interactions
  with probability at least $1-\frac{2n^4}{\exp(\sqrt{n}/4)} < 1-1/n^3$.
From now on, we will consider this high probability event of synthetic coin working properly
  and implicitly take unoin bound over it.

During any relevant interaction, any agent in $\lit{lottery}$ mode 
  with probability at least $1/2 - 1/2^8$ observes $0$ and changes its $\id{mode}$ to the tournament.
Hence, the probability that it increases its payoff more than $5 \log{n}$ times
  is at most $(1/2 + 1/2^8)^{5\log{n}} \leq 1/n^4$.
Taking the union bound over all agents gives that the increase in payoffs in the relevant interactions 
  for all agents are less than $5\log{n}$ with probability at least $1-1/n^3$.
During the first $2n$ irrelevant interactions, by the Chernoff bound,
  the probability that a given agent interacts more than $4 \log{n}$ times is at most $1/n^4$.
Taking an union bound, all agents interact at most $4 \log{n}$ times 
  with probability at least $1-1/n^3$.
Even if they increase their $\id{payoff}$ each time,
  in total, with probability at least $1-O(1)/n^3$, all agents will have payoffs at most $9 \log{n}$.

Since every agent stays in $\lit{seeding}$ mode for four interactions, we can find 
  at least $n/2$ agents, who move to $\lit{lottery}$ mode during a relevant interaction.
Consider any one of the $n/2$ agents.
By assumption, the agent will have probability at least $1/2 - 1/2^8$ of finalizing 
  its $\id{payoff}$ and moving to $\lit{tournament}$ mode.
The probability that the payoff of this agent will be larger than $\log{n}/2$ 
  is thus at least $(1/2 - 1/2^8)^{\log{n}/2} \geq 1/n^{0.6}$.
If the payoff is indeed larger, we are done, otherwise, we can find another agent among 
  the $n/2 - \log{n}/2$ whose interactions we have not yet considered, 
  and analogously get that with probability at least $1/n^{0.6}$, 
  it would get a payoff at least $\log{n}/2$.
We can continue this process, and will end up with about $n/\log{n}$ agents, 
  whose interactions were completely independent, and because of the bias, 
  each of them had a probability of at least $1/n^{0.6}$ of getting a larger payoff than $\log{n}/2$.
If we describe this process as a random variable 
  $\mathrm{Bin} \left( \frac{n}{\log{n}}, \frac{1}{n^{0.6}} \right)$ with expectation $n^{0.4}/\log{n}$, 
  we get by the Chernoff Bound that the probability of no node getting $\geq \log{n}/2$ payoff  
  is extremely low (in particular, less than $1/n^3$).

Again, considering all the high probability events from above, we know that the 
  maximum payoff in the system is between $\log{n}/2$ and $9\log{n}$.
Consider any fixed payoff $k$ in this interval, and let us say it is the maximum.
Then, any agent that reaches this payoff, has to flip $0$, but they might flip
  $1$ with probability at least $1/2 - 1/2^8$.
Thus, the probability that at least $12 \log{n}$ agents will stop exactly 
  at payoff $k$ is at most $(1/2 + 1/2^8)^{5 \log{n}} \leq 1/n^4$.
Taking the union bound over at most $9 \log{n} < n$ payoffs, 
  and the above high probability events, we get that with probability at most $1-O(1)/n^3$,
  at most $5 \log{n}$ agents will have the maximum payoff, which will 
  be between $\log{n} / 2$ and $9 \log{n}$.

It only remains to prove that after $O(n \log{n})$ interactions, 
  all agents will be in the tournament stage with probability $1-O(1)/n^3$,
  which is a standard high probability coupon collector argument.
Recall that after $2n$ irrelevant interactions, with probability at least $1-1/n^3$
  synthetic coins work properly for the next $n^4$ interactions,
  and thus, each agent that has not yet moved to the tournament stage
  and interacts, has probability of at least $1/2 - 1/2^8$ of doing so.
Thus, after $O(n \log{n})$ interactions, by the Chernoff bound,
  a given agent will interact $O(\log{n})$ times with very high probability
  and move to the tournament stage at least during 
  one of these $O(\log{n})$ interactions also with very high probability.
Taking union bound over all agents and these events completes the proof.     
\end{proof}
\begin{lemma}
\label{lem:stage2}
with probability at least $1-O(1)/n^3$,
  only one contender reaches level $\ell = 3 \log{n} / \log \log n$,
  and it takes at most $O(n \log^{5.3} \log \log n)$ interactions 
  for a new level to be reached up to $\ell$. 
\end{lemma}
\begin{proof}
By our assumption on $m$, it holds that $\frac{m}{\log{m}} \geq \frac{3\log{n}}{\log{\log{n}}}$,
  so this level can always be reached. 
We consider the high probability case in~\lemmaref{lem:stage1}, 
  which occurs with probability $\geq 1-O(1)/n^3$. 
Hence, we need to prove that the probability that 
  more than one competitor reaches level $\ell$ is also at most $O(1)/n^3$.
 
Consider some competitor $v$ which just increased the maximum level among competitors in the system. 
Until some other competitor reaches the same level, 
  $v$ will turn every interaction partner into its minion. 
Furthermore, as in epidemic spreading, these minions will also turn their interaction partners 
  into minions of the highest level contender $v$. 
Let the $\id{payoff}$ and $\id{level}$ pair of this competitor be $u$. 
We call a node whose pair also at least $u$ an \emph{up-to-date} node; 
  the node is \emph{out-of-date} otherwise.
Initially, only the contender $v$ that reached the maximum level is up-to-date.

We will show that if in some configuration $x < n$ nodes are up-to-date,
  after a \emph{phase} of $\frac{16 n (n-1) \log{n}}{4x(n-x)}$ interactions, 
  at least $x+1$ nodes will be up-to-date with probability at least $1-\frac{1}{n^5}$.
Up-to-date nodes may never become out-of-date.
On the other hand, an out-of-date node becomes up-to-date itself after an interaction with 
  any up-to-date node.
If we have $x$ up-to-date nodes, 
  in each interaction, the probability that an out-of-date node interacts with an up-to-date node 
  increasing the number of up-to-date nodes to $x+1$, is $\frac{2x(n-x)}{n(n-1)}$,
  which we denote by $1/\alpha$.
To probability that this never happens during $\frac{16 n(n-1)\log{n}}{4x(n-x)} = 8 \alpha \log{n}$
  interactions is then $(1 - 1 / \alpha)^{8 \alpha \log{n}} \leq 1 - 1/n^5$. 

\noindent An Union Bound over at most $n$ phases gives that with probability at least $1 - 1/n^4$, 
  after at most $$\sum_{x = 1}^{n-1} \frac{16 n (n-1) \log{n}}{4x(n-x)} 
  \leq \frac{16 (n-1)\log{n}}{4} \sum_{x = 1}^{n-1} \left( \frac{1}{x} 
  + \frac{1}{n-x}\right) \leq 16 n\log^2 n$$ 
  rounds, all nodes will have value at least $u$.
Taking an union bound over all possible levels $\ell < \log{n} < n$,
  and all previous events, we get that with probability at least $1-O(1)/n^3$, 
  once a contender reaches some level, unless some other contender reaches the same level 
  within the next $16 n \log^2{n}$ interactions,
  the original node will turn every other node into minions and become a stable leader.

Once some contender has increased the maximum level, 
  it needs to observe between $4.2 \log{\log{n}}$ and $4.2 (\log \log n + \log 9 + 1)$  
  consecutive ones to increase a level,
  as with high probability we knew that the payoff was between $\log{n} / 2$ and 
  $9 \log n$ and we set the phase size to $4.2 (\log{\id{payoff}} + 1)$. 
Notice that all nodes that did not get a maximum payoff are not considered as contenders,
  because they will all be eliminated by maximum payoff contenders within $O(n \log {n})$ interactions.

A given node at each iteration observes coin value one with probability at least $(1/2 - 1/2^8)$, 
so the probability that a stretch of $4.2 \log{\log{n}} + O(1)$ consecutive interactions results 
  in a level increase is at least $\Theta(1/\log^{4.3} n)$.
If we consider an interval containing $O(n \log^{5.3}{n} \log\log n)$ interactions,
  then the agent increases the level in expectation $O(\log^{5.3}{n} / \log^{4.3} n) = O(\log n)$
  times and by the Chernoff bound, at least once with high probability.

For the rest of the argument, we will bound the probability that any of the other contenders, 
  whose number is at most $5 \log n$ by~\lemmaref{lem:stage1}, 
  will be able to increment their level during an arbitrary stretch of 
  $6 \log{\log{n}} + 12$ consecutive interactions. 
This probability will be low, and therefore it is likely that the process 
  will terminate after a level increase. 

More precisely, once a new level is reached after a level increment, 
  the nodes have $16 n \log^2{n}$ interactions to increment to the same level,
  or they will soon all become minions.
To do so, they should all have at least one iteration of observing at least $4.2 \log{\log{n}}$ 
  consecutive ones.

Hence, there can be at most $5 \log{n} \cdot \Theta(\log^2{n} / \log{\log{n}})$ such interaction intervals 
  and each interaction interval has probability at most 
  $(1 / 2 + 1 / 2^8)^{4.2\log\log{n}} \leq 1/\log^4{n}$ of success.
Hence, by taking sufficiently large $n$, we can make the 
  expectation of the number of successful iterations be less than $1/\log{n}$.

Hence, the probability that there is no second survivor among contenders at each level 
  (which would correspond to a stable leader being elected) is at most $1/\log{n}$, 
  every time the maximum level is incremented.
The probability that this does not happen for all 
  $\ell = \frac{3\log{n}}{\log{\log{n}}}$ levels is then at most
  $(\frac{1}{\log{n}})^{\ell} \leq \frac{1}{n^3}$.
\end{proof}
\begin{lemma}
\label{lem:electexp}
The expected parallel time until stabilization is $O(\log^{5.3} n \log \log n)$.
\end{lemma}
\begin{proof}
To get the same bound on expected parallel time as in the with high probability case, 
  we should incorporate  the expected time in the low probability event.
The lottery stage takes expected $O(n \log n)$ interactions by standard coupon collector argument
  (see the argument in~\lemmaref{lem:stage1}).

During competition, as we saw in~\lemmaref{lem:stage2}, with high probability, 
  a new level is reached within every $O(n \log^{5.3}{n}\log \log n)$ interactions, 
  and with more than constant probability a single contender remains after each level.
Therefore, in the high probability case, only constantly many levels will be used 
  in expectation and expected parallel time will be $O(\log^{5.3}n \log \log n)$.

Otherwise, during competition, recall that non-minions can always eliminate each other 
  in direct interactions comparing their payoffs, levels, and the coin as a tie-breaker.
So, for any given two non-minion nodes $x$ and $y$, in every interaction, 
  there is a probability of at least $O(1/n^2)$ that they meet, and one eliminates
  each other for certain if they had different coin values.
If not, then with probability at least $1/n$, one of the nodes, say $x$, interacts with some 
  other node in this interaction, and then immediately afterward, interacts with $y$,
  this time with different coin values.
Hence, in every two iterations, with probability at least $O(1/n^3)$, 
  the number of contenders decreases by at least one.
The expected number of interactions until one remains is $O(n^4)$,
  thus parallel time $O(n^3)$ with probability at most $O(1/ n^3)$,
  which gives negligible expectation $O(1)$ on parallel time in this low probability case.
\end{proof}

\end{document}